\documentclass[a4paper,UKenglish,cleveref,autoref,dvipsnames,thm-restate]{lipics-v2021}

\pdfoutput=1 
\hideLIPIcs  

\usepackage{mdframed}
\usepackage{bussproofs}
\usepackage{tikz-cd}
\usepackage{csquotes}

\colorlet{HighlightColor}{Yellow!30}

\title{E-unification for Second-Order Abstract Syntax}

\author{Nikolai {Kudasov}}{Innopolis University, Universitetskaya 1, Innopolis, Tatarstan Republic, Russia}{n.kudasov@innopolis.ru}{https://orcid.org/0000-0001-6572-7292}{}

\authorrunning{N. Kudasov}

\Copyright{Nikolai Kudasov}



\acknowledgements{I thank Nikolay Shilov for encouragement and fruitful discussions about the contents of the paper. I also thank Alexander Chichigin, Violetta Sim, Timur Fayzrahmanov, Alexey Stepanov, and Dale Miller for discussions of early versions of this work. I thank Alexander Chichigin, Timur Fayzrahmanov, and Ikechi Ndukwe for proofreading the paper. Finally, I thank the anonymous reviewers of FSCD 2023 for their valuable and detailed comments.}

\nolinenumbers 

\EventEditors{Marco Gaboardi and Femke van Raamsdonk}
\EventNoEds{2}
\EventLongTitle{8th International Conference on Formal Structures for Computation and Deduction (FSCD 2023)}
\EventShortTitle{FSCD 2023}
\EventAcronym{FSCD}
\EventYear{2023}
\EventDate{July 3--6, 2023}
\EventLocation{Rome, Italy}
\EventLogo{}
\SeriesVolume{260}
\ArticleNo{6}

\DeclareMathOperator{\meq}{\stackrel{?}{=}}

\newcommand{\metavar}[2]{{\normalfont\color{RoyalBlue}{\textsc{#1}{#2}}}}
\newcommand{\type}[1]{{\color{purple}{#1}}}
\newcommand{\con}[1]{{\color{blue}{\mathsf{#1}}}}
\newcommand{\metasubst}[1]{{\color{Plum}{#1}}}

\newcommand{\defcon}{{\con{F}}}
\newcommand{\m}{{\metavar{\textsc{m}}{}}}
\newcommand{\app}{{\con{app}}}
\newcommand{\abs}{{\con{abs}}}
\newcommand{\subst}{{\con{subst}}}
\newcommand{\named}{{\con{named}}}
\newcommand{\ord}{{\mathsf{ord}}}

\newcommand{\defemph}[1]{\emph{\textbf{#1}}}

\newtheorem*{theorem*}{Theorem}
\newtheorem*{definition*}{Definition}

\begin{document}

\maketitle

\begin{abstract}
Higher-order unification (HOU) concerns unification of (extensions of) $\lambda$-calculus and can be seen as an instance of equational unification ($E$-unification) modulo $\beta\eta$-equivalence of $\lambda$-terms. We study equational unification of terms in languages with arbitrary variable binding constructions modulo arbitrary second-order equational theories. Abstract syntax with general variable binding and parametrised metavariables allows us to work with arbitrary binders without committing to $\lambda$-calculus or use inconvenient and error-prone term encodings, leading to a more flexible framework. In this paper, we introduce $E$-unification for second-order abstract syntax and describe a unification procedure for such problems, merging ideas from both full HOU and general $E$-unification. We prove that the procedure is sound and complete.
\end{abstract}

\section{Introduction}

Higher-order unification (HOU) is a process of solving symbolic equations with functions. Consider the following equation in untyped $\lambda$-calculus that we want to solve for $\m$:
\begin{align}
  \m\;g\;(\lambda z. z\;a) \meq g\;a
\end{align}

A solution to this problem (called a unifier) is the substitution $\metasubst{\theta} = [\m \mapsto \lambda x. \lambda y. y\;x]$. Indeed, applying $\metasubst{\theta}$ to the equation we get $\beta$-equivalent terms on both sides:
\[
  \metasubst{\theta}(\m\;g\;(\lambda z. z\;a))
  = (\lambda x. \lambda y. y \; x)\;g \; (\lambda z. z\; a)
  \equiv_{\beta} (\lambda y. y\; g)\;(\lambda z. z\; a)
  \equiv_{\beta} (\lambda z. z\; a)\;g
  \equiv_{\beta} g\;a
  = \metasubst{\theta}(g\;a)
\]

Higher-order unification has many applications, including type checking~\cite{MazzoliAbel2016} and automatic theorem proving in higher-order logics~\cite{MillerNadathur2012}.
In general, HOU is undecidable~\cite{Goldfarb1981} and searching for a unifier can be rather expensive without non-trivial optimizations.
For some problems, a decidable fragment is sufficient to solve for.
For instance, Miller's higher-order pattern unification~\cite{Miller1991} and its variations~\cite{Gundry2013,ZilianiSozeau2015} are often used for dependent type inference.

Traditionally, HOU algorithms consider only one binder~--- $\lambda$-abstraction. A common justification is an appeal to Higher-Order Abstract Syntax (HOAS)~\cite{PfenningElliott1988}:
\begin{displayquote}[{Nipkow and Prehofer~\cite[Section~1]{NipkowPrehofer1998}}]
  It is well-known that $\lambda$-abstraction is general enough to represent quantification in formulae,
  abstraction in functional programs, and many other variable-binding constructs \cite{PfenningElliott1988}.
\end{displayquote}

However, HOAS has received some criticism from both programming language implementors and formalisation researchers, who argue that HOAS and its variants~\cite{Chlipala2008,WashburnWeirich2008} have some practical issues~\cite{Kmett2015_SoH, Cockx2021_blog}, such as being hard to work under binders, having issues with general recursion~\cite{Kmett2008_CR}, and lacking a formal foundation~\cite{FioreSzamozvancev2022}.

Fiore and Szamozvancev~\cite{FioreSzamozvancev2022} argue that existing developments for formalising, reasoning about, and implementing languages with variable bindings \enquote{offer some relief, however at the expense of inconvenient and error-prone term encodings and lack of formal foundations}. Instead, they suggest to consider \emph{second-order abstract syntax}~\cite{FioreHur2010}, that is, abstract syntax with variable binding and parametrised metavariables. Indeed, Fiore and Szamozvancev~\cite{FioreSzamozvancev2022} use second-order abstract syntax to generate metatheory in Agda for languages with variable bindings.

In this paper, we develop a mechanisation of equational reasoning for second-order abstract syntax. We take inspiration in existing developments for HOU and $E$-unification. Although we cannot directly reuse all HOU ideas that rely heavily on the syntax of $\lambda$-calculus, we are still able to adapt many of them, since second-order abstract syntax provides \emph{parametrised metavariables} which are similar to \emph{flex} terms in HOU.

\subsection{Related Work}

To the best of our knowledge, there does not exist a mechanisation of equational reasoning for second-order abstract syntax. Thus, we compare our approach with existing HOU algorithms that encompass equational reasoning. Snyder's higher-order $E$-unification~\cite{Snyder1990} extends HOU with first-order equational theories. Nipkow and Prehofer~\cite{NipkowPrehofer1998} study higher-order rewriting and (higher-order) equational reasoning. As mentioned, these rely on $\lambda$-abstraction and a HOAS-like encoding to work with other binding constructions. In contrast, we work with arbitrary binding constructions modulo a second-order equational theory.

Dowek, Hardin, and Kirchner~\cite{DowekHardinKirchner2000} present higher-order unification as first-order $E$-unification in $\lambda\sigma$-calculus (a variant of $\lambda$-calculus with explicit substitutions) modulo $\beta\eta$-reduction.
Their idea is to use explicit substitutions and de Bruijn indices so that metavariable substitution cannot result in name capture and reduces to \emph{grafting} (first-order substitution).
In this way, algorithms for first-order $E$-unification (such as \emph{narrowing}) can be applied.
Kirchner and Ringeissen~\cite{KirchnerRingeissen1997} develop that approach for higher-order equational unification with first-order axioms.
In our work, parametrised metavariables act in a similar way to metavariables with explicit substitutions in $\lambda\sigma$-calculus.
While it should be possible to encode second-order equations as first-order equations in $\sigma$-calculus (with explicit substitution, but without $\lambda$-abstraction and application), it appears that this approach requires us to also encode rules of our unification procedure.

As some equational theories can be formulated as term rewriting systems, a line of research combining rewrite systems and type systems exists, stemming from the work of Tannen~\cite{Tannen1988}, which extends simply-typed $\lambda$-calculus with higher-order rewrite rules.
Similar extensions exist for the Calculus of Constructions~\cite{BarbaneraFernandezGeuvers1997,Walukiewicz-Chrzkaszcz2003,Stehr2005partI,Stehr2005partII} and $\lambda\Pi$-calculus~\cite{CousineauDowek2007}. 
Cockx, Tabareau, and Winterhalter~\cite{CockxTabareauWinterhalter2021} introduce Rewriting Type Theory (RTT) which is an extension of Martin-L\"of Type Theory with (first-order) rewrite rules.
Chrz\k{a}szcz and Walukiewicz-Chrz\k{a}szcz~\cite{Chrzkaszcz-Walukiewicz-Chrzkaszcz2007} discuss how to extend Coq with rewrite rules.
Cockx~\cite{Cockx2020} reports on a practical extension of Agda with higher-order non-linear rewrite rules, based on the same ideas as RTT~\cite{CockxTabareauWinterhalter2021}.
Rewriting is especially useful in proof assistants that compare types (and terms) through \emph{normalisation by evaluation} (NbE) rather than higher-order unification.
Contrary to type theories extended with rewrite rules, our approach relies on simply-typed syntax, but allows for an arbitrary second-order equational theory, enabling unification even in the absence of a confluent rewriting system.

Kudasov~\cite{Kudasov2022} implements higher-order (pre)unification and dependent type inference in Haskell for an intrinsically scoped syntax using so-called \emph{free scoped monads} to generate the syntax of the object language from a data type describing node types. Such a definition is essentially a simplified presentation of second-order abstract syntax. Kudasov's pre-unification procedure contains several heuristics, however no soundness or completeness results are given in the preprint.

\subsection{Contribution}

The primary contribution of this paper is the introduction of $E$-unification for second-order abstract syntax and a sound and complete unification procedure. The rest of the paper is structured as follows:
\begin{itemize}
  \item In \cref{section:second-order-syntax}, we briefly revisit second-order abstract syntax, equational logic, and term rewriting \`a la Fiore and Hur~\cite{FioreHur2010}.
  \item In \cref{section:e2-unification}, we generalise traditional $E$-unification concepts of an $E$-unification problem and an $E$-unifier for a set of second-order equations $E$.
  \item In \cref{section:unification-procedure}, we define the unification procedure that enumerates solutions for any given $E$-unification problem and prove it sound.
  \item In \cref{section:completeness}, we prove completeness of our unification procedure, taking inspiration from existing research on $E$-unification and HOU.
  \item Finally, we discuss some potential pragmatic changes for a practical implementation as well as limitations of our approach in \cref{section:discussion}.
\end{itemize}

\section{Second-Order Abstract Syntax}
\label{section:second-order-syntax}

In this section, we recall second-order abstract syntax, second-order equational logic, and second-order term rewriting of Fiore and Hur~\cite{FioreHur2010}. 


\subsection{Second-Order Terms}

We start by recalling a notion of second-order signature, which essentially contains information about the syntactic constructions (potentially with variable bindings) of the object language.

\begin{remark}
  There are at least two subtly different ways to look at a syntax with bound variables.
  In particular, a term $\lambda x. t$ can be seen as
  \begin{enumerate}
    \item A binder $\lambda x$ that introduces a variable $x$, and a term $t$ in an extended context; in a uniform syntax, we then present $\lambda$-abstraction with a family of functional symbols $\con{abs}_x(t)$;
    \item A symbol $\lambda$ applied to a \emph{scoped term} $x.t$ where $(x.)$ introduces a new variable $x$ in which then can be used in the term $t$; in a uniform syntax, we then present $\lambda$-abstraction with a single functional symbol $\con{abs}(x.t)$ which takes a scoped term as its argument.
  \end{enumerate}
  Following Fiore and Hur we go with the second interpretation.
\end{remark}

A \defemph{second-order signature} \cite[Section~2]{FioreHur2010} $\Sigma = (T, O, |-|)$ is specified by a set of types $T$, a set of operators\footnote{In literature on $E$-unification, authors use the term \emph{functional symbol} instead.} $O$, and an arity\footnote{We follow the terminology of Fiore and Hur.} function $|-| : O \to (T^* \times T)^* \times T$.
For an operator $\defcon \in O$, we write $\defcon : (\type{\overline{\sigma_1}}.\type{\tau_1},\ldots,\type{\overline{\sigma_n}}.\type{\tau_n}) \to \type{\tau}$ when $|\defcon| = ((\type{\overline{\sigma_1}}, \type{\tau_1}), \ldots, (\type{\overline{\sigma_n}}, \type{\tau_n}), \type{\tau})$. Intuitively, this means that an operator $\defcon$ takes $n$ arguments each of which binds $n_i = |\type{\overline{\sigma_i}}|$ variables of types $\type{\sigma_{i,1}}, \ldots, \type{\sigma_{i,n_i}}$ in a term of type $\type{\tau_i}$.

For the rest of the paper, we assume an ambient signature $\Sigma$, unless otherwise stated.

A \defemph{typing context} \cite[Section~2]{FioreHur2010} $\Theta \mid \Gamma$ consists of metavariable typings $\Theta$ and variable typings $\Gamma$. Metavariable typings are parametrised types: a metavariable of type $[\type{\sigma_1}, \ldots, \type{\sigma_n}]\type{\tau}$, when parametrised by terms of type $\type{\sigma_1}, \ldots, \type{\sigma_n}$, will yield a term of type $\type{\tau}$.
We will write a centered dot ($\cdot$) for the empty (meta)variable context.

For example, this context has a metavariable $\m$ with two parameters and variables $x, y$:
$\Theta \mid \Gamma = (\metavar{m}{} : [\type{\sigma},\type{\sigma \Rightarrow \tau}]\type{\tau} \mid x : \type{\sigma \Rightarrow \tau}, y : \type{\sigma})$.

\begin{figure}
  \begin{mdframed}
  \begin{prooftree}
    \AxiomC{\phantom{$\Theta$}} 
    \noLine
    \UnaryInfC{$x : \type{\tau} \in \Gamma$}
    \RightLabel{variables}
    \UnaryInfC{$\Theta \mid \Gamma \vdash x : \type{\tau}$}
    \DisplayProof\quad\quad
    \AxiomC{$\m : [\type{\sigma_1}, \ldots, \type{\sigma_n}]\type{\tau} \in \Theta$}
    \noLine
    \UnaryInfC{$\text{for all $i = 1, \ldots, n$} \quad
      \Theta \mid \Gamma \vdash t_i : \type{\sigma_i}$}
    \RightLabel{metavariables}
    \UnaryInfC{$\Theta \mid \Gamma \vdash \m[t_1, \ldots, t_n] : \type{\tau}$}
  \end{prooftree}
  \begin{prooftree}
    \AxiomC{$\defcon : (\type{\overline{\sigma_1}}.\type{\tau_1}, \ldots, \type{\overline{\sigma_n}}.\type{\tau_n}) \to \type{\tau}$}
    \noLine
    \UnaryInfC{$\text{for all $i = 1, \ldots, n$} \quad
      \Theta \mid \Gamma, \overline{x_i} : \type{\overline{\sigma_i}} \vdash t_i : \type{\tau_i}$}
    \RightLabel{operators}
    \UnaryInfC{$\Theta \mid \Gamma \vdash \defcon(\overline{x_1}.t_1, \ldots, \overline{x_n}.t_n) : \type{\tau}$}
  \end{prooftree}
  \end{mdframed}
  \caption{Second-order terms in context.}
  \label{figure:second-order-terms}
\end{figure}

\begin{definition}[{\cite[Section~2]{FioreHur2010}}]
  A judgement for typed \defemph{terms} in context $\Theta \mid \Gamma \vdash - : \type{\tau}$ is defined by the rules in \cref{figure:second-order-terms}. Variable substitution on terms is defined in a usual way, see \cite[Section~2]{FioreHur2010} for details.

  Let $\Theta = (\metavar{m}{_i} : [\type{\overline{\sigma_i}}]\type{\tau_i})^{i \in \{1,\ldots,n\}}$, and consider a term $\Theta \mid \Gamma \vdash t : \type{\tau}$, and for all $i \in \{1, \ldots, n\}$ a term in extended\footnote{Here we slightly generalise the definition of Fiore and Hur by allowing arbitrary extension of context to $\Gamma, \Delta$ in the resulting term. This is useful in particular when $\Gamma$ is empty. See~\cref{definition:instantiate-axiom}.} context $\Xi \mid \Gamma, \Delta, \overline{z_i} : \overline{\sigma_i} \vdash t_i : \type{\tau_i}$.
  Then, \defemph{metavariable substitution} $t[\metavar{m}{_i}[\overline{z_i}] \mapsto t_i]^{i \in \{1, \ldots, n\}}$ is defined recursively on the structure of $t$:
\begin{align*}
  x[\metavar{m}{_i}[\overline{z_i}] \mapsto t_i]^{i \in \{1, \ldots, n\}}
  &= x \span \\
  \metavar{m}{_k}[\overline{s}][\metavar{m}{_i}[\overline{z_i}] \mapsto t_i]^{i \in \{1, \ldots, n\}}
  &= t_k[\overline{z_k} \mapsto \overline{s[\metavar{m}{_i}[\overline{z_i}] \mapsto t_i]^{i \in \{1, \ldots, n\}}}] & \\
  &\text{when $k \in \{1, \ldots, n\}$ and $|\overline{s}| = |\overline{z_k}|$} & \\
  \metavar{n}{}[\overline{s}][\metavar{m}{_i}[\overline{z_i}] \mapsto t_i]^{i \in \{1, \ldots, n\}}
  &= \metavar{n}{}[\overline{s[\metavar{m}{_i}[\overline{z_i}] \mapsto t_i]^{i \in \{1, \ldots, n\}}}]
  &\text{when $\metavar{n}{} \not\in \{\metavar{m}{_1}, \ldots, \metavar{m}{_n}\}$} \\
  \defcon(\overline{\overline{x}.s})[\metavar{m}{_i}[\overline{z_i}] \mapsto t_i]^{i \in \{1, \ldots, n\}}
  &= \defcon(\overline{\overline{x}.s[\metavar{m}{_i}[\overline{z_i}] \mapsto t_i]^{i \in \{1, \ldots, n\}}})
  \span
\end{align*}

We write $\metasubst{\theta} : \Theta \mid \Gamma \to \Xi \mid \Gamma, \Delta$ for a substitution $\metasubst{\theta} = [\metavar{m}{_i}[\overline{z_i}] \mapsto t_i]^{i \in \{1, \ldots, n\}}$.

When both $\Gamma$ and $\Delta$ are empty, we write $\metasubst{\theta} : \Theta \to \Xi$ as a shorthand for $\metasubst{\theta} : \Theta \mid \cdot \to \Xi \mid \cdot$.

For single metavariable substitutions in a larger context we will omit the metavariables that map to themselves. That is, we write $[\metavar{m}{_k}[\overline{z}] \mapsto t_k] : \Theta \mid \Gamma \to \Xi \mid \Gamma, \Delta$ to mean that $t_i = \metavar{m}{_i}[\overline{z_i}]$ for all $i \not= k$.
\end{definition}

\begin{example} Here are some examples of $\lambda$-terms in second-order syntax:
  \begin{itemize}
    \item identity function ($\lambda x. x$): $ \con{abs}(x. x)$
    \item Church numeral $c_3 = \lambda s. \lambda z. s\;(s\;(s\;z))$: $\con{abs}(s. \con{abs}(z. \con{app}(s, \con{app}(s, \con{app}(s, z)))))$
    \item swapping function ($\lambda p. \langle \pi_2\;p, \pi_1\;p \rangle$): $\con{abs}(p. \con{pair}(\con{snd}(p), \con{fst}(p)))$
    \item application of swap ($(\lambda p. \langle \pi_2\;p, \pi_1\;p \rangle)\;\langle a, b \rangle$): $\con{app}(\con{abs}(p. \con{pair}(\con{snd}(p), \con{fst}(p))), \con{pair}(a, b))$
  \end{itemize}
\end{example}

\subsection{Second-Order Equational Logic}

We now define second-order equational presentations and rules of second-order logic, following Fiore and Hur~\cite[Section 5]{FioreHur2010}. This provides us with tools for reasoning modulo second-order equational theories, such as $\beta\eta$-equivalence of $\lambda$-calculus.

An \defemph{equational presentation} \cite[Section~5]{FioreHur2010} is a set of axioms each of which is a pair of terms in context.

\begin{example}
  Terms of simply-typed $\lambda$-calculus are generated with a family of operators (for all $\type{\sigma}, \type{\tau}$)~---
  $\abs^{\type{\sigma},\type{\tau}}
      : \type{\sigma}.\type{\tau} \to (\type{\sigma \Rightarrow \tau})$
  and
  $\app^{\type{\sigma},\type{\tau}}
      : (\type{\sigma \Rightarrow \tau}, \type{\sigma}) \to \type{\tau}$.
  And equational presentation for simply-typed $\lambda$-calculus is given by a family of axioms:
  \begin{align*}
    \metavar{m}{} : [\type{\sigma}]\type{\tau}, \metavar{n}{} : []\type{\sigma}
      \mid \cdot
      &\vdash \con{app}(\con{abs}(x. \metavar{m}{}[x]), \metavar{n}{}[])
      \equiv \metavar{m}{}[\metavar{n}{}[]] : \type{\tau}
      \tag{$\beta$} \\
    \metavar{m}{} : []\type{\sigma \Rightarrow \tau}
      \mid \cdot
      &\vdash \con{abs}(x. \con{app}(\metavar{m}{}[], x))
      \equiv \metavar{m}{}[] : \type{\sigma \Rightarrow \tau}
      \tag{$\eta$}
  \end{align*}
\end{example}

Note that the types here do not depend on the context, so it makes sense to only allow equating terms of the same type. This is in contrast to dependently typed systems, where terms can have different (but equivalent) types. 

An equational presentation $E$ generates a \defemph{second-order equational logic} \cite[Fig.~2]{FioreHur2010}. Rules for second-order equational logic are given in \cref{figure:second-order-equational-logic}.

\begin{figure}
  \begin{mdframed}
  \small
  \begin{prooftree}
    \AxiomC{$(\Theta \mid \Gamma \vdash s \equiv t : \type{\tau}) \in E$}
    \RightLabel{axiom}
    \UnaryInfC{$\Theta \mid \Gamma \vdash s \equiv_E t : \type{\tau}$}
    \DisplayProof\quad\quad\quad
    \AxiomC{$\Theta \mid \Gamma \vdash t : \type{\tau}$}
    \RightLabel{refl}
    \UnaryInfC{$\Theta \mid \Gamma \vdash t \equiv_E t : \type{\tau}$}
  \end{prooftree}
  \begin{prooftree}
    \AxiomC{$\Theta \mid \Gamma \vdash s \equiv_E t : \type{\tau}$}
    \RightLabel{sym}
    \UnaryInfC{$\Theta \mid \Gamma \vdash t \equiv_E s : \type{\tau}$}
    \DisplayProof\quad\quad\quad
    \AxiomC{$\Theta \mid \Gamma \vdash s \equiv_E t : \type{\tau}$}
    \AxiomC{$\Theta \mid \Gamma \vdash t \equiv_E u : \type{\tau}$}
    \RightLabel{trans}
    \BinaryInfC{$\Theta \mid \Gamma \vdash s \equiv_E u : \type{\tau}$}
  \end{prooftree}
  \begin{prooftree}
    \AxiomC{$\metavar{m}{_1} : [\type{\overline{\sigma_1}}]\type{\tau_1}, \ldots, \metavar{m}{_n} : [\type{\overline{\sigma_n}}]\type{\tau_n} \mid \Gamma \vdash s \equiv_E t : \type{\tau}$}
    \noLine
    \UnaryInfC{for all $i \in \{1, \ldots, n\} \quad \Theta \mid \Delta, \overline{x_i} : \type{\overline{\sigma_i}} \vdash s_i \equiv_E t_i : \type{\tau_i}$}
    \RightLabel{subst}
    \UnaryInfC{$\Theta \mid \Gamma, \Delta \vdash s[\metavar{m}{_1}[\overline{x_1}] \mapsto s_1, \ldots, \metavar{m}{_n}[\overline{x_n}] \mapsto s_n] \equiv_E t[\metavar{m}{_1}[\overline{x_1}] \mapsto t_1, \ldots, \metavar{m}{_n}[\overline{x_n}] \mapsto t_n] : \type{\tau}$}
  \end{prooftree}
  \end{mdframed}
  \caption{Rules of the second-order equational logic.}
  \label{figure:second-order-equational-logic}
\end{figure}

In their paper, Fiore and Hur note that metavariables with zero parameters are equivalent to regular variables. Indeed, we can \defemph{parametrise} every term $\Theta \mid \Gamma \vdash t : \type{\tau}$ to yield a term $\Theta, \widehat{\Gamma} \mid \cdot \vdash \widehat{t} : \type{\tau}$ where for $\Gamma = (x_1 : \type{\sigma_1}, \ldots, x_n : \type{\sigma_n})$ we have
\begin{align*}
  \widehat{\Gamma} = (\metavar{x}{_1} : []\type{\sigma_1}, \ldots, \metavar{x}{_n} : []\type{\sigma_n}) \quad\quad\quad
  \widehat{t} = t[x_1 \mapsto \metavar{x}{_1}[], \ldots, x_n \mapsto \metavar{x}{_n}[]]
\end{align*}

Applying parametrisation to an equational presentation $E$ yields a set of parametrised equations $\widehat{E}$. Note that the following are equivalent:
\begin{align*}
  \Theta \mid \Gamma \vdash s \equiv_E t : \type{\tau}
  \quad\quad\quad\text{iff}\quad\quad\quad
  \Theta, \widehat{\Gamma} \mid \cdot \vdash \widehat{s} \equiv_{\widehat{E}} \widehat{t} : \type{\tau}
\end{align*}

Thus, from now on, we assume that axioms have empty variable context.

\subsection{Second-Order Term Rewriting}

Finally, for the proof of completeness in \cref{section:completeness}, it will be helpful to rely on chains of term rewrites rather than derivation trees of equality modulo $E$.
Fiore and Hur introduce the \emph{second-order term rewriting} relation~\cite[Section 8]{FioreHur2010}.

An equational presentation $E$ generates a \defemph{second-order term rewriting relation} $\longrightarrow_E$ \cite[Fig.~4]{FioreHur2010}. We write $s \stackrel{*}{\longrightarrow}_E t$ if there is a sequence of terms $u_1, \ldots, u_n$ such that $s = u_1 \longrightarrow_E \ldots \longrightarrow_E u_n = t$. We write $s \longleftrightarrow_E t$ if either $s \longrightarrow_E t$ or $t \longrightarrow_E s$. We write $s \stackrel{*}{\longleftrightarrow}_E t$ if there is a sequence of terms $u_1, \ldots, u_n$ such that $s = u_1 \longleftrightarrow_E \ldots \longleftrightarrow_E u_n = t$.

Since we only care about substitutions of metavariables in axioms (variable context is empty), a simplified version of the rules is given in \cref{figure:second-order-term-rewriting}.

\begin{figure}
  \begin{mdframed}
  \begin{prooftree}
    \AxiomC{$(\metavar{m}{_1} : [\type{\overline{\sigma_1}}]\type{\tau_1}, \ldots, \metavar{m}{_k} : [\overline{\sigma_k}]\type{\tau_k} \mid \cdot \vdash l \equiv r : \type{\tau}) \in E$}
    \noLine
    \UnaryInfC{$\Theta \mid \Gamma, \overline{x_i} : \type{\overline{\sigma_i}} \vdash t_i : \type{\tau_i} : $\quad for $i \in \{1, \ldots, k\}$}
    \UnaryInfC{$\Theta \mid \Gamma \vdash l[\metavar{m}{_i}[\overline{z_i}] \mapsto t_i]^{i \in \{1, \ldots, k\}} \longrightarrow r[\metavar{m}{_i}[\overline{z_i}] \mapsto t_i]^{i \in \{1, \ldots, k\}} : \type{\tau}$}
  \end{prooftree}
  \begin{prooftree}
    \AxiomC{$\metavar{m}{} : [\type{\overline{\sigma}}]\type{\tau}
      \quad\quad \Theta \mid \Gamma \vdash s_i \longrightarrow t_i : \type{\sigma_i}$}
    \UnaryInfC{$\Theta \mid \Gamma \vdash \metavar{m}{}[\ldots, s_i, \ldots] \longrightarrow \metavar{m}{}[\ldots, t_i, \ldots] : \type{\tau}$}
  \end{prooftree}
  \begin{prooftree}
    \AxiomC{$\con{F} : (\type{\overline{\sigma_1}}.\type{\tau_1}, \ldots, \type{\overline{\sigma_n}}.\type{\tau_n}) \to \type{\tau}
    \quad\quad \Theta \mid \Gamma, \overline{x_i} : \type{\overline{\sigma_i}} \vdash s_i \longrightarrow t_i : \type{\tau_i}$}
    \UnaryInfC{$\Theta \mid \Gamma \vdash \con{F}(\ldots, \overline{x_i}.s_i, \ldots) \longrightarrow \con{F}(\ldots, \overline{x_i}.t_i, \ldots) : \type{\tau}$}
  \end{prooftree}
  \end{mdframed}
  \caption{Rules of the second-order term rewriting (simplified). In the second and third rules, the subterms under ($\ldots$) are kept unchanged, so only one subterm is rewritten per rule.}
  \label{figure:second-order-term-rewriting}
\end{figure}

An important result of Fiore and Hur is that of soundness and completeness of second-order term rewriting~\cite[Section~8]{FioreHur2010}:
  $
    \Theta \mid \Gamma \vdash s \equiv_E t : \type{\tau}
    \text{\; iff \;}
    \Theta \mid \Gamma \vdash s \stackrel{*}{\longleftrightarrow}_E t : \type{\tau}
    $.

\subsection{Example: $\lambda\mu$-calculus}

Parigot's $\lambda\mu$-calculus~\cite{Parigot1992} features $\mu$-abstraction and the following reduction rule:
\begin{align}
  (\mu \beta. u)\;v
    \longrightarrow \mu \beta. u [[\beta]w \mapsto [\beta](w\;v)]
    \tag{structural reduction}
\end{align}

This particular rule is interesting because of a non-trivial binder and because the substitution on the right hand side demands replacing every subterm of $u$ of the form $[\beta]w$ with subterm $[\beta](w\;v)$. Such substitution can be internalized as an explicit substitution operator, or modelled using $\lambda$-abstraction and traditional substitution (see \cref{example:lambda-mu-signature}).

Below we give a presentation of untyped $\lambda\mu$-calculus in second-order abstract syntax:

\begin{example}
  \label{example:lambda-mu-signature}
  For the untyped $\lambda\mu$-calculus and with explicit substitution operator for names, a signature consists of:
  \begin{enumerate}
    \item a set of types $T = \{\type{\star}, \type{\mathsf{Name}}\}$
    \item a set of four operators:
      \begin{align*}
        \abs &: \type{\star}.\type{\star} \to \type{\star}
          \tag{$\lambda$-abstraction: $\lambda x. t$} \\
        \app &: (\type{\star}, \type{\star}) \to \type{\star}
          \tag{application: $t_1\;t_2$} \\
        \con{mu} &: \type{\mathsf{Name}}.\type{\star} \to \type{\star}
          \tag{$\mu$-abstraction: $\mu \alpha. t$} \\
        \named &: (\type{\mathsf{Name}}, \type{\star}) \to \type{\star}
          \tag{named term: $[\alpha]t$} \\
        \subst &: (\type{\mathsf{Name}}.\type{\star}, \type{\star}) \to \type{\star}
        \tag{name substitution: $[[\alpha]w \mapsto [\alpha](w\; v)]t$}
      \end{align*}
  \end{enumerate}

  The following equalities correspond to the structural reduction and to the definition of name substitution:
  \begin{align}
    \metavar{m}{_1} : [\type{\mathsf{Name}}]\type{\star}, \metavar{m}{_2} : \type{\star} \mid \cdot
      &\vdash \app(\con{mu}(\alpha.\metavar{m}{_1}[\alpha]), \metavar{m}{_2}[]) 
      &\equiv\;& \con{mu}(\alpha.\subst(\alpha.\metavar{m}{_1}[\alpha], \metavar{m}{_2}[])) \notag \\
    \m : \type{\star} \mid x : \type{\star}
      &\vdash \subst(\alpha.x, \m[])
      &\equiv\;& x \notag \\
    \metavar{m}{_1} : [\type{\star}]\type{\star}, \metavar{m}{_2} : \type{\star} \mid \cdot
      &\vdash \subst(\alpha.\abs(x.\metavar{m}{_1}[x]), \metavar{m}{_2}[]) 
      &\equiv\;& \abs(x.\subst(\alpha.\metavar{m}{_1}[x], \metavar{m}{_2}[])) \notag \\
    \metavar{m}{_1}, \metavar{m}{_2}, \metavar{m}{_3} : \type{\star} \mid \cdot
      &\vdash \subst(\alpha.\app(\metavar{m}{_1}[], \metavar{m}{_2}[]), \metavar{m}{_3}[]) 
      &\equiv\;& \app(\subst(\alpha.\metavar{m}{_1}[], \metavar{m}{_3}[]),\subst(\alpha.\metavar{m}{_2}[], \metavar{m}{_3}[])) \notag \\
    \metavar{m}{_1}, \metavar{m}{_2} : \type{\star} \mid \cdot
      &\vdash \subst(\alpha.\named(\alpha, \metavar{m}{_1}[]), \metavar{m}{_2}[]) 
      &\equiv\;& \named(\alpha, \app(\subst(\alpha.\metavar{m}{_1}[], \metavar{m}{_2}[]), \metavar{m}{_2}[])) \notag \\
      \metavar{m}{_1}, \metavar{m}{_2} : \type{\star} \mid \beta : \type{\mathsf{Name}}
      &\vdash \subst(\alpha.\named(\beta, \metavar{m}{_1}[]), \metavar{m}{_2}[]) 
      &\equiv\;& \named(\beta, \subst(\alpha.\metavar{m}{_1}[], \metavar{m}{_2}[])) \notag \\
    \metavar{m}{_1}, \metavar{m}{_2} : \type{\star} \mid \cdot
      &\vdash \subst(\alpha.\con{mu}(\beta.\metavar{m}{_1}[\beta]), \metavar{m}{_2}[]) 
      &\equiv\;& \con{mu}(\beta.\subst(\alpha.\metavar{m}{_1}[\beta], \metavar{m}{_2}[])) \notag
  \end{align}
\end{example}

\section{E-unification with Second-Order Equations}
\label{section:e2-unification}

In this section, we formulate the equational unification problem for second-order abstract syntax, describe what constitutes a solution to such a problem and whether it is complete. We also recognise a subclass of problems in solved form, i.e. problems that have an immediate solution. For the most part, this is a straightforward generalisation of standard concepts of $E$-unification \cite{GallierSnyder1989}.

\begin{definition}
  A \defemph{second-order constraint} $\Theta \mid \Gamma_\exists, \Gamma_\forall \vdash s \meq t : \type{\tau}$ is a pair of terms in a context, where variable context is split into two components: $\Gamma = (\Gamma_\exists, \Gamma_\forall)$.
\end{definition}

The idea is that $\Gamma_\exists$ contains variables that we need to solve for (free variables), while $\Gamma_\forall$ contains variables that we cannot substitute (bound variables). Metavariables are always treated existentially, so we do not split metavariable context.
Similarly to equational representations, we can parametrise (a set of) constraints, yielding $\Theta, \widehat{\Gamma}_\exists \mid \Gamma_\forall \vdash s \meq t : \type{\tau}$.
Thus, from now on, we will assume $\Gamma = \Gamma_\forall$ (i.e. $\Gamma_\exists$ is empty) for all constraints.

\begin{example}
  Assume $\type{\alpha} = \type{\sigma \Rightarrow \tau}, \type{\beta} = \type{(\sigma \Rightarrow \tau) \Rightarrow \tau}$.
  The following are equivalent:
  \begin{enumerate}
    \item \label{constraint-lambda} For all $g : \type{\alpha}, a : \type{\sigma}$, find $m : \type{\alpha \Rightarrow \beta \Rightarrow \tau}$ such that $m \; g \; (\lambda z. z \; a) = g \; a$.
    \item \label{constraint-simple}
      $\cdot \mid m : \type{\alpha \Rightarrow \beta \Rightarrow \tau}, g : \type{\alpha}, a : \type{\sigma}
      \vdash \con{app}(\con{app}(m, g), \con{abs}(z. \con{app}(z, a)))
      \meq \con{app}(g, a) : \type{\tau}$
    \item \label{constraint-param}
      $\metavar{m}{} : []\type{\alpha \Rightarrow \beta \Rightarrow \tau} \mid g : \type{\alpha}, a : \type{\sigma}
      \vdash \con{app}(\con{app}(\m[], g), \con{abs}(z. \con{app}(z, a)))
      \meq \con{app}(g, a) : \type{\tau}$
    \item \label{constraint-param-equiv}
      $\metavar{m}{} : [\type{\alpha}, \type{\beta}]\type{\tau} \mid g : \type{\alpha}, a : \type{\sigma}
      \vdash \metavar{m}{}[g, \con{abs}(z. \con{app}(z, a))]
      \meq \con{app}(g, a) : \type{\tau}$
  \end{enumerate}

  Here, \cref{constraint-simple} is a direct encoding of \cref{constraint-lambda} as a second-order constraint.
  \cref{constraint-param} is a parametrised version of \cref{constraint-simple}.
  \cref{constraint-param-equiv} is equivalent to \cref{constraint-param} modulo $\beta$-equality,
  witnessed by metasubstitutions $[\m[] \mapsto \con{abs}(x.\con{abs}(y.\con{app}(x, y)))]$ and $[\m[x, y] \mapsto \con{app}(\con{app}(\m[], x), y)]$.

\end{example}

\begin{definition}
  Given an equational presentation $E$, an \defemph{$E$-unification problem} $\langle \Theta, S \rangle$ is a finite set $S$ of second-order constraints in a shared metavariable context $\Theta$. We present an $E$-unification problem as a formula of the following form:
  \[
    \exists (\metavar{m}{_1} : [\type{\overline{\sigma_1}}]\type{\tau_1}, \ldots, \metavar{m}{_n} : [\type{\overline{\sigma_n}}]\type{\tau_n}).
    (\forall (\overline{z_1} : \type{\overline{\rho_1}}). s_1 \meq t_1 : \type{\tau_1})
    \land \ldots \land
    (\forall (\overline{z_k} : \type{\overline{\rho_k}}). s_k \meq t_k : \type{\tau_k})
  \]
\end{definition}

Here is an example of $E$-unification problem for simply-typed $\lambda$-terms:
 \[
   \exists (\metavar{m}{} : [\type{\sigma \Rightarrow \tau}, \type{(\sigma \Rightarrow \tau) \Rightarrow \tau}]\type{\tau}).
   \;\; \forall (g : \type{\sigma \Rightarrow \tau}, y : \type{\sigma}). 
   \;\; \metavar{m}{}[g, \con{abs}(x. \con{app}(x, y))]
     \meq \con{app}(g, y) : \type{\tau}
 \]

\begin{definition}
  A metavariable substitution $\metasubst{\xi} : \Theta \to \Xi$ is called an \defemph{$E$-unifier} for an $E$-unification problem $\langle \Theta, S \rangle$ if for all constraints $(\Theta \mid \Gamma_\forall \vdash s \meq t : \type{\tau}) \in S$ we have
  \[
    \Xi \mid \Gamma_\forall \vdash \metasubst{\xi} s \equiv_E \metasubst{\xi} t : \type{\tau}
  \]
  We write $U_E(S)$ for the set of all $E$-unifiers for $\langle \Theta, S \rangle$.
\end{definition}

\begin{example}
  Consider unification problem $\langle \Theta, S \rangle$ for the simply-typed $\lambda$-calculus:
  \begin{align*}
    \Theta &= \metavar{m}{} : [\type{\sigma \Rightarrow \tau}, \type{(\sigma \Rightarrow \tau) \Rightarrow \tau}]\type{\tau} \\
    S &= \{ \Theta \mid g : \type{\sigma \Rightarrow \tau}, y : \type{\sigma} \vdash \metavar{m}{}[g, \con{abs}(x. \con{app}(x, y))] \meq \con{app}(g, y) : \type{\tau} \}
  \end{align*}
  Substitution $[\metavar{m}{}[z_1, z_2] \mapsto \con{app}(z_2, z_1)] : \Theta \to \cdot$
  is an $E$-unifier for $\langle \Theta, S \rangle$.
\end{example}

\subsection{Unification Problems in Solved Form}

Here, we recognise a class of trivial unification problems. The idea is that a constraint that looks like a metavariable substitution can be uniquely unified. A unification problem can be unified as long as substitutions for constraints are sufficiently disjoint. More precisely:

\begin{definition}
  An $E$-unification problem $\langle \Theta, S \rangle$ is in \defemph{solved form} when $S$ consists only of constraints of the form
  $ \Theta, \metavar{m}{} : [\type{\overline{\sigma}}]\type{\tau} \mid \Gamma_\forall \vdash \metavar{m}{}[\overline{z}] \meq t : \type{\tau} $
  such that
  \begin{enumerate}
    \item $\overline{z} : \type{\overline{\sigma}} \subseteq \Gamma_\forall$\quad
      (parameters of $\metavar{m}{}$ are \emph{distinct} variables from $\Gamma_\forall$)
    \item $\Theta \mid \overline{z} : \type{\overline{\sigma}} \vdash t : \type{\tau}$\quad
      ($\metavar{m}{}$ and variables not occurring in $\overline{z}$ do not occur in $t$)
    \item all constraints have distinct metavariables on the left hand side
  \end{enumerate}
\end{definition}

\begin{example}
  Let $\Theta = (\metavar{m}{} : [\type{\sigma}, \type{\sigma}]\type{\sigma})$. Then
  \begin{enumerate}
    \item $\{ \Theta \mid x : \type{\sigma}, y : \type{\sigma} \vdash
      \metavar{m}{}[y, x] \meq \con{app}(\con{abs}(z. x), y) : \type{\sigma} \}$
      is in solved form;
    \item $\{ \Theta \mid x : \type{\sigma}, y : \type{\sigma} \vdash
      \metavar{m}{}[x, x] \meq \con{app}(\con{abs}(z. x), y) : \type{\sigma} \}$
      is not in solved form, since parameters of $\m$ are not \emph{distinct} variables and also since variable $y$ occurs on the right hand side, but does not occur in parameters of $\m$;
    \item $\{ \Theta \mid f : \type{\sigma \Rightarrow \sigma}, y : \type{\sigma} \vdash
      \metavar{m}{}[y, \con{app}(f, y)] \meq \con{app}(f, y) : \type{\sigma} \}$
      is not in solved form, since second parameter of $\m$ is not a variable;
  \end{enumerate}
\end{example}

\begin{proposition}
  \label{proposition:solved-form-unifier}
  An $E$-unification problem $\langle \Theta, S \rangle$ in solved form has an $E$-unifier.
\end{proposition}
\begin{proof}
  Assume $\Theta = \{ \Theta \mid \Gamma_i \vdash \metavar{m}{_i}[\overline{z_i}] \meq t_i : \type{\tau_i} \}^{i \in \{1, \ldots, n\}}$. Let
  $ \metasubst{\xi_S} = [\metavar{m}{_i}[\overline{z_i}] \mapsto t_i]^{i \in \{1, \ldots, n\}} $.
  Note that $\metasubst{\xi_S}$ is a well formed metasubstitution since, by assumption, each $\overline{z_i}$ is a sequence of distinct variables, $t_i$ does not reference other variables or $\metavar{m}{_i}$, and each metavariable $\metavar{m}{_i}$ is mapped only once in $\metasubst{\xi_S}$. Applying $\metasubst{\xi_S}$ to each constraint we get trivial constraints, which are satisfied by reflexivity: $\Theta \mid \Gamma_i \vdash t_i \equiv_E t_i : \type{\tau_i}$. Thus, $\metasubst{\xi_S}$ is an $E$-unifier for $\langle \Theta, S \rangle$.
\end{proof}

Later, we will refer to the $E$-unifier constructed in the proof of \cref{proposition:solved-form-unifier} as $\metasubst{\xi_S}$.

\subsection{Comparing E-unifiers}

In general, a unification problem may have multiple unifiers. Here, we generalise the usual notion of comparing $E$-unifiers~\cite{GallierSnyder1989} to the second-order abstract syntax using the \emph{subsumption} order, leading to a straightforward generalisation of the ideas of \emph{the most general unifier} and \emph{a complete set of unifiers}. We do not consider generalising \emph{essential unifiers}~\cite{HocheSzabo2006,SzaboSiekmannHoche2016} or \emph{homeomorphic embedding}~\cite{SzaboSiekmann2021}, although these might constitute a prospective future work.

\begin{definition}
  Two metavariable substitutions $\metasubst{\theta}, \metasubst{\xi} : \Theta \to \Xi$ are said to be \defemph{equal modulo $E$} (notated $\metasubst{\theta} \equiv_E \metasubst{\xi}$) if for all metavariables $\metavar{m}{} : [\type{\overline{\sigma}}]\type{\tau} \in \Theta$, any context $\Gamma$, and any terms $\Theta \mid \Gamma \vdash t_i : \type{\sigma_i}$ (for all $i \in \{ 1, \ldots, n \}$) we have
  \[
    \Xi \mid \Gamma \vdash \metasubst{\theta} \metavar{m}{}[t_1, \ldots, t_n] \equiv_E \metasubst{\xi} \metavar{m}{}[t_1, \ldots, t_n] : \type{\tau}
  \]

  We say that $\metasubst{\theta}$ is \defemph{more general modulo $E$ than} $\metasubst{\xi}$ (notated $\metasubst{\theta} \preccurlyeq_E \metasubst{\xi}$) when there exists a substitution $\metasubst{\eta} : \Xi \to \Xi$ such that $\metasubst{\eta \circ \theta} \equiv_E \metasubst{\xi}$.
\end{definition}

  Empty substitution is more general than any substitution.
  A more interesting example may be found in $\lambda$-calculus. Let
  \begin{align*}
    \metasubst{\theta_1} &= [\metavar{m}{}[x, y] \mapsto \con{app}(\metavar{n}{}[x], y)] \\
    \metasubst{\theta_2} &= [ \metavar{m}{}[x, y] \mapsto \con{app}(\con{abs}(z. x), y)
    , \metavar{n}{}[x] \mapsto \con{abs}(z. x) ] \\
    \metasubst{\theta_3} &= [ \metavar{m}{}[x, y] \mapsto x
    , \metavar{n}{}[x] \mapsto \con{abs}(z. x) ]
  \end{align*}
  Then, $\metasubst{\theta_1} \preccurlyeq_E \metasubst{\theta_2}$, $\metasubst{\theta_2} \equiv_E \metasubst{\theta_3}$, and $\metasubst{\theta_1} \preccurlyeq_E \metasubst{\theta_3}$ (witnessed by $[\metavar{n}{}[x] \mapsto \con{abs}(z. x)] \metasubst{\; \circ \; \theta_1} \equiv_E \metasubst{\theta_3}$).

\begin{proposition}
  If $\metasubst{\theta} \equiv_E \metasubst{\xi}$ then for any $E$-unification problem $\langle \Theta, S \rangle$ we have $\metasubst{\theta} \in U_E(S)$ iff $\metasubst{\xi} \in U_E(S)$.
\end{proposition}
\begin{proof}
  For each constraint $\Theta \mid \Gamma_\forall \vdash s \meq t : \type{\tau}$, by induction on the structure of $s$ and $t$ it is straightforward to show that
  $\Xi \mid \Gamma \vdash \metasubst{\theta} s \equiv_E \metasubst{\theta} t : \type{\tau}$
    iff
    $\Xi \mid \Gamma \vdash \metasubst{\xi} s \equiv_E \metasubst{\xi} t : \type{\tau}$.
\end{proof}
\begin{corollary}
  If $\metasubst{\theta} \preccurlyeq_E \metasubst{\xi}$ and $\metasubst{\theta} \in U_E(S)$ then $\metasubst{\xi} \in U_E(S)$.
\end{corollary}

Not all substitutions can be compared.
Consider untyped lambda calculus with $\type{\star}$ being the type of any term.
Let $\Theta = (\metavar{m}{} : [\type{\star}, \type{\star}]\type{\star})$ and
\begin{align*}
  \metasubst{\theta} &= [ \metavar{m}{}[z_1, z_2] \mapsto \con{app}(z_2, z_1) ] \\
  \metasubst{\xi}    &= [ \metavar{m}{}[z_1, z_2] \mapsto \con{app}(z_1, \con{app}(z_2, \con{abs}(z.z)))]
\end{align*}
None of these substitutions is more general modulo equational theory $E$ of untyped $\lambda$-calculus than the other. At the same time, both are $E$-unifiers for the problem
\[ \exists (\metavar{m}{} : [\type{\star}, \type{\star}]\type{\star}). \;
     \forall (g : \type{\star}, y : \type{\star}). \;
       \metavar{m}{}[g, \con{abs}(x. \con{app}(x, y))] \meq \con{app}(g, y) : \type{\star}
   \]

\subsection{Complete Sets of E-unifiers}

While there is sometimes more than one solution to an $E$-unification problem, we may often hope to collect several sufficiently general unifiers into a single set:

\begin{definition}
  \label{definition:csu}
  Given an $E$-unification problem $\langle \Theta, S \rangle$, a (minimal)
  \defemph{complete set of $E$-unifiers} for $\langle \Theta, S \rangle$ (notated $\mathsf{CSU}_E(S)$) is a subset of $U_E(S)$ such that
  \begin{enumerate}
    \item (completeness) for any $\metasubst{\eta} \in U_E(S)$ there exists $\metasubst{\theta} \in \mathsf{CSU}_E(S)$ such that $\metasubst{\theta} \preccurlyeq_E \metasubst{\eta}$;
    \item (minimality) for any $\metasubst{\theta}, \metasubst{\xi} \in \mathsf{CSU}_E(S)$ if $\metasubst{\theta} \preccurlyeq_E \metasubst{\xi}$ then $\metasubst{\theta} = \metasubst{\xi}$.
  \end{enumerate}
\end{definition}

We reserve the notation $\mathsf{CSU}_E(S)$ to refer to minimal complete sets of $E$-unifiers (i.e. satisfying both conditions).

\begin{example}
  \label{example:csu-untyped}
  The $E$-unification problem $\langle \Theta, S \rangle$ in untyped $\lambda$-calculus has an infinite $\mathsf{CSU}_E(S)$:
  \begin{align*}
    \langle \Theta, S \rangle = \exists (\metavar{m}{} : [\type{\star}, \type{\star}]\type{\star}). \; 
    \forall (g : \type{\star}, y : \type{\star}). \;
    \metavar{m}{}[g, \con{abs}(x. \con{app}(x, y))] \meq \con{app}(g, y) : \type{\star}
    \span
    \\
    \mathsf{CSU}_E(S) = \{
    & [ \metavar{m}{}[z_1, z_2] \mapsto \con{app}(z_2, z_1) ] , \\
    & [ \metavar{m}{}[z_1, z_2] \mapsto \con{app}(z_1, \con{app}(z_2, \con{abs}(x.x))) ] , \\
    & [ \metavar{m}{}[z_1, z_2] \mapsto \con{app}(\con{app}(z_2, \con{abs}(x. \con{abs}(f. \con{app}(f, x)))), z_1) ]
    , \ldots \}
  \end{align*}
\end{example}

\begin{proposition}
  For any two minimal complete sets of $E$-unifiers $\mathsf{CSU}^1_E(S)$ and $\mathsf{CSU}^2_E(S)$, there exists a bijection $f : \mathsf{CSU}^1_E(S) \longleftrightarrow \mathsf{CSU}^2_E(S)$ such that
\[
  \forall \metasubst{\theta} \in \mathsf{CSU}^1_E(S).\quad \metasubst{\theta} \equiv_E f(\metasubst{\theta})
\]
\end{proposition}


Thus, $\mathsf{CSU}_E(S)$ is unique up to a bijection modulo $E$,
so from now on we will refer to \emph{the} complete set of $E$-unifiers.

\begin{definition}
  When the complete set of $E$-unifiers $\mathsf{CSU}_E(S)$ is a singleton set, then we refer to its element as \defemph{the most general $E$-unifier} of $S$ (notated $\mathsf{mgu}_E(S)$).
\end{definition}

\begin{example}
  \label{example:mgu}
  Consider this $E$-unification problem $\langle \Theta, S \rangle$ in simply-typed $\lambda$-calculus:
  \[
    \exists (\metavar{m}{} : [\type{\sigma \Rightarrow \tau}, \type{(\sigma \Rightarrow \tau) \Rightarrow \tau}]\type{\tau}). \;
    \forall (g : \type{\sigma \Rightarrow \tau}, y : \type{\sigma}). \;
    \metavar{m}{}[g, \con{abs}(x. \con{app}(x, y))] \meq \con{app}(g, y) : \type{\tau}
  \]
  For this problem the most general $E$-unifier exists:
  $\mathsf{mgu}_E(S) = [ \metavar{m}{}[z_1, z_2] \mapsto \con{app}(z_2, z_1) ]$.
  This example differs from~\cref{example:csu-untyped} as here we work in simply-typed lambda calculus.
\end{example}

\begin{proposition}
  If $\langle \Theta, S \rangle$ is an $E$-unification problem in solved form, then $\mathsf{mgu}_E(S) \equiv_E \metasubst{\xi_S}$.
\end{proposition}
\begin{proof}
  It is enough to check that for any $E$-unifier $\metasubst{\theta} \in U_E(S)$ we have $\metasubst{\xi_S} \preccurlyeq_E \metasubst{\theta}$.
  Observe that $\metasubst{\theta} \equiv_E \metasubst{\theta \circ \xi_S}$ since for any constraint $(\Theta \mid \Gamma_\forall \vdash M[\overline{z}] \meq t : \type{\tau}) \in S$ such that $\metavar{m}{} : [\overline{\type{\sigma}}]\type{\tau} \in \Theta$, any context $\Gamma$, and any terms $\Theta \mid \Gamma \vdash t_i : \type{\sigma_i}$ (for all $i \in \{ 1, \ldots, n \}$) we have
  \begin{align*}
    \Xi \mid \Gamma
    &\vdash \metasubst{\theta} \metavar{m}{}[\overline{t}]
    \equiv_E \metasubst{\theta} t[\overline{z} \mapsto \overline{t}]
    \equiv_E \metasubst{\theta} (\metasubst{\xi_S} \metavar{m}{}[\overline{z}])[\overline{z} \mapsto \overline{t}]
    \equiv_E \metasubst{\theta} (\metasubst{\xi_S} \metavar{m}{}[\overline{t}]) : \type{\tau}
  \end{align*}
\end{proof}

%

\section{Unification Procedure}
\label{section:unification-procedure}

In this section, we introduce a unification procedure to solve arbitrary $E$-unification problems over second-order abstract syntax. We show that the procedure is sound at the end of this section, and we devote~\cref{section:completeness} for the completeness result. 

Our unification procedure has features inspired by classical $E$-unification and HOU algorithms. For the equational part, we took inspiration from the complete sets of transformations for general (first-order) $E$-unification of Gallier and Snyder~\cite{GallierSnyder1989}. For unification of metavariables, we took inspiration from Huet's higher-order pre-unification~\cite{Huet1975} and Jensen and Pietrzykowski's procedure~\cite{JensenPietrzykowski1976}. Some key insights from the recent work by Vukmirovic, Bentkamp, and Nummelin~\cite{VukmirovicBentkampNummelin2021} give us the opportunity to improve the algorithm further, however, we are not attempting to achieve an \emph{efficient} $E$-unification for second-order abstract syntax in this paper.

Note that we cannot directly reuse HOU ideas in our procedure, since we do not have full $\lambda$-calculus at our disposal. Instead we only have parametrised metavariables $\m[t_1, \ldots, t_n]$ which are analogous to applications of variables in HOU ($m \; t_1 \; \ldots \; t_n$). Still, we can adapt some ideas if they do not rely on normalisation or specific syntax of $\lambda$-calculus. For other ideas, we introduce simpler, yet more general versions. This allows us to preserve completeness, perhaps, sacrificing some efficiency, making the search space larger. While we believe it is possible to optimise our procedure to have virtually the same running time for unification problems in $\lambda$-calculus as HOU algorithms mentioned above, we leave such optimisations for future work.

To produce the unification procedure we follow and generalise some of the common steps that can be found in literature on HOU and first-order $E$-unification:
\begin{enumerate}
  \item Classify substitutions that will constitute partial solutions for certain classes of constraints. The idea is that an overall solution will emerge as a composition of partial solutions.
  \item Define transition rules that make small steps towards a solution.
  \item Determine when to stop (succeed or fail).
  \item If possible, organize rules in a proper order, yielding a unification procedure.
\end{enumerate}

\subsection{Bindings}

Now we define different elementary substitutions that will serve as partial solutions for some constraints in our unification procedure. Here, we generalise a list of bindings collected by Vukmirovic, Bentkamp, and Nummelin~\cite{VukmirovicBentkampNummelin2021}. From that list, Huet-style projection (also known as \emph{partial binding} in HOU literature) is not used. Instead, imitation for axioms and JP-style projection bindings cover all substitutions that can be generated by Huet-style projection bindings\footnote{Note, that Huet-style projection cannot be formulated in pure second-order abstract syntax as it explicitly relies on $\con{abs}$ and $\con{app}$. Thus, in $E$-unification we can recover such projections only by using axioms in some form. Kudasov~\cite{Kudasov2022} implements a heuristic that resembles a generalisation of Huet-style projections. We leave proper generalisations for future work.}. We also use a simplified version of iteration binding here, again, since it generates all necessary bindings when considered together with generalised imitation binding.

\begin{definition}
  We define the following types of bindings $\metasubst{\zeta}$:

\begin{description}
  \item[JP-style projection for $\m$.]
    If $\metavar{m}{} : [\type{\sigma_1}, \ldots, \type{\sigma_k}]\type{\tau}$
    and $\type{\sigma_i} = \type{\tau}$ then \\
    $\metasubst{\zeta} = [\metavar{m}{}[\overline{z}] \mapsto z_i]$
    is a JP-style projection binding

  \item[Imitation for $\m$.]
    If $\metavar{m}{} : [\type{\sigma_1}, \ldots, \type{\sigma_k}]\type{\tau}$,
    $\con{F} : (\type{\overline{\alpha_1}}.\type{\beta_1}, \ldots, \type{\overline{\alpha_n}}.\type{\beta_n}) \to \type{\tau}$
    and $\metavar{m}{_i} : [\type{\sigma_1}, \ldots, \type{\sigma_k}, \type{\overline{\alpha_i}}]\type{\beta_i}$ for all $i$,
    \\
    $\metasubst{\zeta} = [\metavar{m}{}[\overline{z}] \mapsto \con{F}(\overline{x_1}.\metavar{m}{_1}[\overline{z},\overline{x_1}], \ldots, \overline{x_n}.\metavar{m}{_n}[\overline{z},\overline{x_n}])]$
    is an imitation binding

  \item[Elimination for $\m$.]
    If $\metavar{m}{} : [\type{\sigma_1}, \ldots, \type{\sigma_k}]\type{\tau}$
    and $1 \leq j_1 < j_2 < \ldots < j_{n-1} < j_n \leq k$ such that
    $\metavar{e}{} : [\type{\sigma_{j_1}, \ldots, \type{\sigma_{j_n}}}]\type{\tau}$
    then \\
    $\metasubst{\zeta} = [\metavar{m}{}[\overline{z}] \mapsto \metavar{e}{}[z_{j_1}, \ldots, z_{j_n}]]$
    is a (parameter) elimination binding

  \item[Identification of $\m$ and $\metavar{n}{}$.]
    If $\metavar{m}{} : [\type{\sigma_1}, \ldots, \type{\sigma_k}]\type{\tau}$,
       $\metavar{n}{} : [\type{\nu_1}, \ldots, \type{\nu_l}]\type{\tau}$,
       $\metavar{i}{} : [\type{\sigma_1}, \ldots, \type{\sigma_k}, \type{\nu_1}, \ldots, \type{\nu_l}]\type{\tau}$, \\
       $\metavar{m}{_i} : [\type{\sigma_1}, \ldots, \type{\sigma_k}]\type{\nu_i}$ for all $i \in \{1, \ldots, l\}$, and
       $\metavar{n}{_j} : [\type{\nu_1}, \ldots, \type{\nu_l}]\type{\sigma_j}$ for all $j \in \{1, \ldots, k\}$ then \\
       $\metasubst{\zeta} = [\metavar{m}{}[\overline{z}] \mapsto \metavar{i}{}[\overline{z}, \metavar{m}{_1}[\overline{z}], \ldots, \metavar{m}{_l}[\overline{z}]], \;\;
    \metavar{n}{}[\overline{y}] \mapsto \metavar{i}{}[\metavar{n}{_1}[\overline{y}], \ldots, \metavar{n}{_k}[\overline{y}], \overline{y}]]$
    is an identification binding

  \item[Iteration for $\m$.]
    If $\metavar{m}{} : [\type{\sigma_1}, \ldots, \type{\sigma_k}]\type{\tau}$,
       $\con{F} : (\type{\overline{\alpha_1}}.\type{\beta_1}, \ldots, \type{\overline{\alpha_n}}.\type{\beta_n}) \to \type{\gamma}$,
       $\metavar{h}{} : [\type{\sigma_1}, \ldots, \type{\sigma_k}, \type{\gamma}]\type{\tau}$, and
       $\metavar{m}{_i} : [\type{\sigma_1}, \ldots, \type{\sigma_k}, \type{\overline{\alpha_i}}]\type{\beta_i}$ for all $i$,
    then \\
    $\metasubst{\zeta} = [\metavar{m}{}[\overline{z}] \mapsto \metavar{h}{}[\overline{z},\con{F}(\overline{x_1}.\metavar{m}{_1}[\overline{z}, \overline{x_1}], \ldots, \overline{x_n}.\metavar{m}{_n}[\overline{z}, \overline{x_n}])]]$
    is an iteration binding

\end{description}

\end{definition}

The iteration bindings allow to combine parameters of a metavariable in arbitrary ways. This is also particularly influenced by the fact that the type $\type{\gamma}$ used in the bindings may be arbitrary. This type of bindings introduce arbitrary branching in the procedure below, so should be used with caution in pragmatic implementations. Intuitively, we emphasize two distinct use cases for the iteration bindings:
\begin{enumerate}
  \item To extract a new term from one or more parameters by application of an axiom. In this case, we use iteration, where the root of one of the sides of an axiom is used as an operator $\con{F}$.
  \item To introduce new variables in scope. In this case, any operator that introduces at least one variable into scope is used in an iteration. This use case is important for the completeness of the procedure. See~\cref{example:iteration-by-scope}.
\end{enumerate}

\subsection{Transition Rules}

We will write each transition rule of the unification procedure in the form
$(\Theta \mid \Gamma_\forall \vdash s \meq t : \type{\tau}) \stackrel{\metasubst{\theta}}{\longrightarrow} \langle \Xi, S \rangle$, where $\metasubst{\theta} : \Theta \to \Xi$ is a metavariable substitution and $S$ is a new set of constraints that is supposed to replace $s \meq t$. We will often write $S$ instead of $\langle \Xi, S \rangle$ when $\Xi$ is understood from context.

We will now go over the rules that will constitute the $E$-unification procedure when put in proper order.
The first two rules are straightforward.

\begin{definition}[delete]
  If a constraint has the same term on both sides, we can \defemph{delete} it:
\begin{align*}
  &(\Theta \mid \Gamma_\forall \vdash t \meq t : \type{\tau})
  \stackrel{\metasubst{\mathsf{id}}}{\longrightarrow} \varnothing
\end{align*}
\end{definition}

\begin{definition}[decompose]
  We define two variants of this rule:
  \begin{enumerate}
    \item
      Let $\con{F} : (\type{\overline{\sigma_1}}.\type{\tau_1}, \ldots, \type{\overline{\sigma_n}}.\type{\tau_n}) \to \type{\tau}$, then we can \defemph{decompose} a constraint with $\con{F}$ on both sides into a set of constraints for each pair of (scoped) subterms:
    \begin{align*}
      &(\Theta \mid \Gamma_\forall \vdash \con{F}(\overline{\overline{x}.t}) \meq \con{F}(\overline{\overline{x}.s}) : \type{\tau})
      \stackrel{\metasubst{\mathsf{id}}}{\longrightarrow}
      \{ \Theta \mid \Gamma_\forall, \overline{x_i} : \type{\overline{\sigma_i}}
      \vdash t_i \meq s_i : \type{\tau_i} \}^{i \in \{1, \ldots, n\}}
    \end{align*}

    \item
      Let $\m : [\type{\sigma_1}, \ldots, \type{\sigma_n}] \type{\tau}$, then we can \defemph{decompose} a constraint with $\m$ on both sides into a set of constraints for each pair of parameters:
    \begin{align*}
      &(\Theta \mid \Gamma_\forall \vdash \m[\overline{t}] \meq \m[\overline{s}] : \type{\tau})
      \stackrel{\metasubst{\mathsf{id}}}{\longrightarrow}
      \{ \Theta \mid \Gamma_\forall
      \vdash t_i \meq s_i : \type{\sigma_i} \}^{i \in \{1, \ldots, n\}}
    \end{align*}
  \end{enumerate}
\end{definition}

\begin{example}
\begin{align*}
  &\Theta \mid \Gamma = (\metavar{m}{} : [\type{\sigma}]\type{\sigma \Rightarrow \sigma} \mid f : \type{\sigma \Rightarrow \sigma}) \\
  &\{\Theta \mid \Gamma
  \vdash \con{abs}(x.\con{app}(\metavar{m}{}[x], x)) \meq \con{abs}(x.\con{app}(f, x)) \} \\
  \stackrel{\metasubst{\mathsf{id}}}{\longrightarrow}
  & \{ \Theta \mid \Gamma, x : \type{\sigma}
\vdash \con{app}(\metavar{m}{}[x], x)) \meq \con{app}(f, x) \} \}
  \tag{decompose} \\
  \stackrel{\metasubst{\mathsf{id}}}{\longrightarrow}
  & \{ \Theta \mid \Gamma, x : \type{\sigma}
  \vdash \metavar{m}{}[x] \meq f, \quad
       \Theta \mid \Gamma, x : \type{\sigma}
  \vdash x \meq x \}
  \tag{decompose} \\
  \stackrel{\metasubst{\mathsf{id}}}{\longrightarrow}
  & \{ \Theta \mid \Gamma, x : \type{\sigma} \vdash \metavar{m}{}[x] \meq f \}
  \tag{delete}
\end{align*}
\end{example}

The next two rules are second-order versions of \emph{imitate} and \emph{project} rules used in many HOU algorithms. The idea is that a metavariable can either imitate the other side of the constraint, or simply project one of its parameters:

\begin{definition}[imitate]
For constraints with a metavariable $\metavar{m}{} : [\type{\overline{\sigma_s}}]\type{\tau}$
and an operator $\con{F} : (\type{\overline{\sigma_1}}.\type{\tau_1}, \ldots, \type{\overline{\sigma_n}}.\type{\tau_n}) \to \type{\tau}$
we can \defemph{imitate} the operator side using an imitation binding (metavariables $\overline{\metavar{t}{}}$ are fresh):
\begin{align*}
  (\Theta \mid \Gamma_\forall \vdash \metavar{m}{}[\overline{s}] \meq \con{F}(\overline{\overline{x}.t}) : \type{\tau})
  \stackrel{\underrightarrow{\;\;[\metavar{m}{}{[\overline{z_s}]} \mapsto \con{F}(\overline{\overline{x}.\metavar{t}{}{[\overline{z_s}, \overline{x}]}})]\;\;}}{\;}
  \{ \Theta \mid \Gamma_\forall \vdash \con{F}(\overline{\overline{x}.\metavar{t}{}{[\overline{s}, \overline{x}]}}) \meq \con{F}(\overline{\overline{x}.t}) : \type{\tau} \}
\end{align*}
\end{definition}

Note that \textbf{(imitate)} can be followed up by an application of the \textbf{(decompose)} rule.

\begin{definition}[project]
  For constraints with a metavariable $\metavar{m}{} : [\type{\overline{\sigma_s}}]\type{\tau}$ and a term $u : \type{\tau}$, if $\type{\sigma_i} = \type{\tau}$ then we can produce a JP-style projection binding for the parameter at position $i$:
\begin{align*}
  (\Theta \mid \Gamma_\forall \vdash \metavar{m}{}[\overline{s}] \meq u : \type{\tau})
  \stackrel{\underrightarrow{\;\;[\metavar{m}{}{[\overline{z}]} \mapsto z_i]\;\;}}{\;}
  \{ \Theta \mid \Gamma_\forall \vdash s_i \meq u : \type{\tau} \}
\end{align*}
\end{definition}

The next rule is concerned with matching one side of a constraint against one side of an axiom.
When matching with an axiom, we need to \emph{instantiate} it to the particular use (indeed, an axiom serves as a schema!). However, it is not sufficient to simply map metavariables of the axiom into fresh metavariables of corresponding types. Since we are instantiating axiom for a particular constraint which may have a non-empty $\Gamma_\forall$, it is important to add all those variables to each of the fresh metavariables\footnote{This is different to $E$-unification with first-order axioms, where metavariables do not carry their own context and can be unified with an arbitrary variable later.}:

\begin{definition}
\label{definition:instantiate-axiom}
  Let $\Gamma_\forall = (\overline{x} : \type{\overline{\alpha}})$
  and $\metasubst{\xi} : \Xi \mid \cdot \to \Theta \mid \Gamma_\forall$. \\
  We say $\metasubst{\xi}$ \defemph{instantiates the axiom}
  $\Xi \mid \cdot \vdash l \equiv r : \type{\tau}$
  \defemph{in context}
  $\Theta \mid \Gamma_\forall$ if
  \begin{enumerate}
    \item for any $(\metavar{m}{_i} : [\type{\overline{\sigma}}]\type{\tau}) \in \Xi$,
      $\metasubst{\xi}$ maps $\metavar{m}{_i}[\overline{t}]$ to $\metavar{n}{_i}[\overline{t}, \overline{x}]$;
    \item $\metavar{n}{_i} = \metavar{n}{_j}$ iff $i = j$ for all $i, j$.
  \end{enumerate}
\end{definition}

\begin{example}
  Let $\metasubst{\xi} = [\metavar{m}{}[z] \mapsto \metavar{m}{_1}[z,g,y], \metavar{n}{}[] \mapsto \metavar{n}{_1}[g,y]]$.
  Then, $\metasubst{\xi}$ instantiates the axiom
  \[
    \metavar{m}{} : [\type{\sigma}]\type{\tau}, \metavar{n}{} : []\type{\sigma}
      \mid \cdot
      \vdash \con{app}(\con{abs}(x.\metavar{m}{}[x]), \metavar{n}{}[]) \equiv \metavar{m}{}[\metavar{n}{}[]] : \type{\tau}
  \]
  in context $\metavar{m}{_1} : [\type{\sigma}, \type{\sigma \Rightarrow \tau}, \type{\sigma}]\type{\tau}, \metavar{n}{_1} : [\type{\sigma \Rightarrow \tau}, \type{\sigma}]\type{\sigma} \mid g : \type{\sigma \Rightarrow \tau}, y : \type{\sigma}$.
\end{example}

\begin{definition}[mutate]
  For constraints where one of the sides matches\footnote{We check that the roots of terms match. Technically, we do not have to perform this check and apply \textbf{(mutate)} rule for any axiom (non-deterministically), since full matching will be performed by the unification procedure.} an axiom in $E$:
\[
  \Xi \mid \cdot \vdash l \equiv r : \type{\tau}
\]
We rewrite the corresponding side (here, $\metasubst{\xi}$ instantiates the axiom in context $\Theta \mid \Gamma_\forall$).
\begin{align*}
  (\Theta \mid \Gamma_\forall \vdash t \meq s : \type{\tau})
  \stackrel{\underrightarrow{\;\;\metasubst{\mathsf{id}}}\;\;}{\;}
  \{ \Theta \mid \Gamma_\forall \vdash t \meq \metasubst{\xi} l : \type{\tau}\}
  \uplus \{ \Theta \mid \Gamma_\forall \vdash \metasubst{\xi} r \meq s : \type{\tau} \}
\end{align*}
\end{definition}

In general, we may rewrite in both directions.
However, it may be pragmatic to choose a single direction to some of the axioms (e.g. $\beta\eta$-reductions), while keeping others bidirectional (e.g. commutativity and associativity axioms).
Note that, unlike previous rules, the \textbf{(mutate)} rule can lead to infinite transition sequences.

The remaining rules deal with constraints with metavariables on both sides.
One rule attempts to unify distinct metavariables:

\begin{definition}[identify]
  When a constraint consists of a pair of distinct metavariables $\metavar{m}{} : [\type{\sigma_1}, \ldots, \type{\sigma_k}]\type{\tau}$ and $\metavar{n}{} : [\type{\gamma_1}, \ldots, \type{\gamma_l}]\type{\tau}$, we can use an identification binding (metavariables $\metavar{i}{}, \overline{\metavar{m}{'}}, \overline{\metavar{n}{'}}$ are fresh):
\begin{align*}
  ( \Theta \mid \Gamma_\forall \vdash \metavar{m}{}[\overline{s}] \meq \metavar{n}{}[\overline{t}] )
  \stackrel{\underrightarrow{\;\;[ \metavar{m}{}[\overline{z}] \mapsto \metavar{i}{}[\overline{z}, \overline{\metavar{m}{'}[\overline{z}]}], \metavar{n}{}[\overline{y}] \mapsto \metavar{i}{}[\overline{\metavar{n}{'}[\overline{y}]}, \overline{y}] ]\;\;}}{\;}
  \{
    \Theta \mid \Gamma_\forall \vdash
    \metavar{i}{}[\overline{s}, \overline{\metavar{m}{'}[\overline{s}]}]
    \meq
    \metavar{i}{}[\overline{\metavar{n}{'}[\overline{u}]}, \overline{u}]
  \}
\end{align*}
\end{definition}

Another rule attempts to unify identical metavariables with distinct lists of parameters:

\begin{definition}[eliminate]
  When a constraint has the same metavariable $\metavar{m}{} : [\type{\sigma_1}, \ldots, \type{\sigma_n}]\type{\tau}$ on both sides and there is a sequence $(j_k)_{k=1}^{n}$ such that $s_{j_k} = t_{j_k}$ for all $k \in \{1, \ldots, n\}$, then we can \defemph{eliminate} every other parameter and leave the remaining terms identical (metavariable $\metavar{e}{}$ is fresh):
\begin{align*}
  ( \Theta \mid \Gamma_\forall \vdash \metavar{m}{}[\overline{s}] \meq \metavar{m}{}[\overline{t}] )
  \stackrel{\underrightarrow{\;\;[ \metavar{m}{}[\overline{z}] \mapsto \metavar{e}{}[z_{j_1}, \ldots, z_{j_n}] ]\;\;}}{\;}
  \varnothing
\end{align*}
\end{definition}

The idea of the final rule is to extend a list of parameters with some combination of those that exist already.
For example, consider constraint $\forall x, y, z. \m[\con{pair}(x, y), z] \meq \metavar{n}{}[x, z]$. It is clear, that if we can work with a pair of $x$ and $y$, then we can work with them individually, since we can extract $x$ using $\con{fst}$ and $y$ using $\con{snd}$. Thus, a substitution $[\m[p, z] \mapsto \metavar{m}{_1}[p, z, \con{fst}(p)]]$ would result in a new constraint $\forall x, y, z. \metavar{m}{_1}[\con{pair}(x, y), z, \con{fst}(\con{pair}(x, y))] \meq \metavar{n}{}[x, z]$. This one can now be solved by applying \textbf{(identify)}, \textbf{(eliminate)}, and \textbf{(decompose)} rules that will lead us to $\forall x, y, z. \con{fst}(\con{pair}(x, y)) \meq x : \type{\sigma}$ which will be processed using \textbf{(mutate)} rule.

\begin{definition}[iterate]
  When a constraint consists of a pair of (possibly, identical) metavariables $\metavar{m}{} : [\type{\sigma_1}, \ldots, \type{\sigma_k}]\type{\tau}$ and $\metavar{n}{} : [\type{\gamma_1}, \ldots, \type{\gamma_l}]\type{\tau}$, we can use an iteration binding (metavariables $\metavar{h}{}, \overline{\metavar{k}{}}$ are fresh):
\begin{align*}
  ( \Theta \mid \Gamma_\forall \vdash \metavar{m}{}[\overline{s}] \meq \metavar{n}{}[\overline{t}] )
  \stackrel{\underrightarrow{\;\;[\metavar{m}{}[\overline{z}] \mapsto \metavar{h}{}[\overline{z},\con{F}(\overline{\overline{x}.\metavar{k}{}[\overline{z}, \overline{x}]})]]\;\;}}{\;}
  \{
    \Theta \mid \Gamma_\forall \vdash
    \metavar{h}{}[\overline{s}, \con{F}(\overline{\overline{x}.\metavar{k}{}[\overline{z}, \overline{x}]})]
    \meq
    \metavar{n}{}[\overline{t}]
  \}
\end{align*}
\end{definition}

The following example demonstrates the importance of iteration by an arbitrary operator to introduce variables into scope:

\begin{example}
  \label{example:iteration-by-scope}
  Consider a unification problem for simply-typed $\lambda$-calculus:
  \begin{align*}
    &\exists
    \m : [\type{\sigma \Rightarrow \sigma \Rightarrow \tau}](\type{\sigma \Rightarrow \tau}) \\
    &\quad
    \forall f : \type{\sigma \Rightarrow \sigma \Rightarrow \sigma \Rightarrow \sigma \Rightarrow \tau}. \\
    &\quad\quad
    \m[\lambda x. \lambda y. f \; x \; y \; x \; x] \meq \m[\lambda x. \lambda y. f \; y \; y \; x \; y]
    : \type{\sigma \Rightarrow \tau}
  \end{align*}
  It has the following $E$-unifier: $\metasubst{\zeta} = [\m[g] \mapsto \lambda z. g \; z \; z]$.
  To construct this unifier from bindings, we start with iteration binding $[\m[g] \mapsto \metavar{i}{}[g, \lambda z. \metavar{m}{_1}[g, z]]]$, introducing the lambda abstraction, which is followed by a projection $[\metavar{i}{}[g, r] \mapsto r]$, which is followed by another iteration (to introduce application), and so on.
\end{example}

Finally, we compile all transition rules into the unification procedure:

\begin{definition}
\label{definition:naive-unification}
  The \defemph{$E$-unification procedure} over an equational presentation $E$ is defined by repeatedly applying the following transitions (non-deterministically) until a stop:
  \begin{enumerate}
    \item If no constraints are left, then stop \textbf{\color{OliveGreen}(succeed)}.
    \item If possible, apply \textbf{(delete)} rule.
    \item If possible, apply \textbf{(mutate)} or \textbf{(decompose)} rule (non-det.).
    \item If there is a constraint consisting of two non-metavariables and none of the above transitions apply, stop \textbf{\color{red}(fail)}.
    \item If there is a constraint $\m[\ldots] \meq \con{F}(\ldots)$, apply \textbf{(imitate)} or \textbf{(project)} rules (non-det.).
    \item If there is a constraint $\m[\ldots] \meq x$, apply \textbf{(project)} rules (non-det.).
    \item If possible, apply \textbf{(identify)}, \textbf{(eliminate)}, or \textbf{(iterate)} rules (non-det.).
    \item If none of the rules above are applicable, then stop \textbf{\color{red}{(fail)}}.
  \end{enumerate}
\end{definition}

Many HOU algorithms \cite{Miller1991,LibalMiller2016} implement a rule (typically called \emph{eliminate}) that allows to eliminate metavariables, when a corresponding constraint is in solved form. Such a rule is not necessary here, as it is covered by a combination of \textbf{(imitate)}, \textbf{(decompose)}, \textbf{(delete)}, \textbf{(identify)}, and \textbf{(eliminate)} rules. However, it simplifies presentation of examples and also serves as a practical optimisation, so we include it as an optional rule:

\begin{definition}[eliminate*]
  When a constraint $C = (\Theta \mid \Gamma_\forall \vdash \m[\overline{z}] \meq u)$ is in solved form, we can eliminate it with a corresponding unifier $\metasubst{\xi_{\{C\}}} = [\m[\overline{z}] \mapsto u]$:
  \begin{align*}
  ( \Theta \mid \Gamma_\forall \vdash \metavar{m}{}[\overline{s}] \meq u )
  \stackrel{\underrightarrow{\;\;[\m[\overline{z}] \mapsto u]\;\;}}{\;}
  \varnothing
\end{align*}
\end{definition}

The \textbf{(eliminate*)} rule should have the same priority as \textbf{(delete)} in the procedure.


\begin{lemma}
\label{lemma:sound-steps}
  In the procedure defined in~\cref{definition:naive-unification}, each step is sound.
  That is,
  if $S \stackrel{\metasubst{\theta}}{\longrightarrow} S'$ is a single-step transition that the procedure takes and $\metasubst{\xi} \in U_E(S')$
  then $\metasubst{\xi \circ \theta} \in U_E(S)$.
\end{lemma}
\begin{proof}
  It is sufficient to show that each step is sound with respect to the constraint it acts upon.
  That is, we consider the step $\{C\} \stackrel{\metasubst{\theta}}{\longrightarrow} S''$ such that $C \in S$ and $S'' \subseteq S'$. By assumption $\metasubst{\xi} \in U_E(S')$ and thus also $\metasubst{\xi} \in U_E(S'')$.
  Note that for any constraint $D \in (S - \{C\})$ we have a corresponding constraint $D' \in (S' - S'')$ such that $D' = \metasubst{\theta}D$. Since $\metasubst{\xi}$ unifies $D'$ it follows that $\metasubst{\xi \circ \theta}$ unifies $D$.
  Thus, it is enough for us to show that $\metasubst{\xi \circ \theta}$ unifies $U_E(\{C\})$.

  We now go over the list of possible steps:
  \begin{itemize}
    \item \textbf{(delete)}: it is clear that any substitution unifies $C$;
    \item \textbf{(decompose)}: since $\metasubst{\xi}$ unifies all subterm pairs in $S''$, it also unifies $C$;
    \item \textbf{(imitate)}, \textbf{(project)}, \textbf{(identify)}, \textbf{(eliminate)}, \textbf{(iterate)}: all of these rules simply make a decision on how to substitute some metavariables (choose $\metasubst{\theta}$) and immediately apply that substitution. So, $S'' = \{\metasubst{\theta}C\}$ and since $\metasubst{\xi}$ unifies $S'$ then $\metasubst{\xi \circ \theta}$ unifies $S$.
    \item \textbf{(mutate)}: let $C = (\Theta \mid \Gamma_\forall \vdash s \meq t : \type{\tau})$ and we mutate according to axiom $(\Xi \mid \cdot \vdash l \equiv r : \type{\tau}) \in E$ with substitution $\metasubst{\zeta}$ instantiating this axiom. By assumption, $\metasubst{\xi}$ unifies both $s \meq \metasubst{\zeta}l$ and $\metasubst{\zeta}r \meq t$. Also, $\Theta \mid \Gamma_\forall \vdash \metasubst{\zeta}l \equiv_E \metasubst{\zeta}r : \type{\tau}$. In this rule, $\metasubst{\theta} = \metasubst{\mathsf{id}}$, and so we can show that $\metasubst{\xi \circ \theta} = \metasubst{\xi}$ unifies $s \meq t$:
  $
    \metasubst{\xi} s \equiv_E \metasubst{\xi}(\metasubst{\zeta}l) \equiv_E \metasubst{\xi}(\metasubst{\zeta}r) \equiv_E \metasubst{\xi} t
  $
  \end{itemize}
\end{proof}

\begin{theorem}
  The procedure defined in ~\cref{definition:naive-unification} is sound. That is, if
  $S \stackrel{\metasubst{\theta_1}}{\longrightarrow} S_1 \stackrel{\metasubst{\theta_2}}{\longrightarrow} \ldots \stackrel{\metasubst{\theta_n}}{\longrightarrow} \varnothing$ is a path produced by the procedure, then $\metasubst{\theta_1 \circ \theta_2 \circ \ldots \circ \theta_n} \in U_E(S)$.
\end{theorem}
\begin{proof}
  Direct corollary of~\cref{lemma:sound-steps}.
\end{proof}

\subsection{Example: simply-typed $\lambda$-calculus with pairs}

Here we present a simple extension of the simply-typed $\lambda$-calculus in second-order syntax:

\begin{example}
  \label{example:lambda-signature}
  A signature for simply-typed $\lambda$-calculus with pairs over a set of base types $A$ consists of:
  \begin{enumerate}
    \item a set of types $A^{\Rightarrow,\times}$ given by
      \begin{prooftree}
        \AxiomC{$\type{\tau} \in A$}
        \UnaryInfC{$\type{\tau} \in A^{\Rightarrow,\times}$}
        \DisplayProof\quad\quad\quad
        \AxiomC{$\type{\sigma}, \type{\tau} \in A^{\Rightarrow,\times}$}
        \UnaryInfC{$\type{\sigma \Rightarrow \tau} \in A^{\Rightarrow,\times}$}
        \DisplayProof\quad\quad\quad
        \AxiomC{$\type{\sigma}, \type{\tau} \in A^{\Rightarrow,\times}$}
        \UnaryInfC{$\type{\sigma \times \tau} \in A^{\Rightarrow,\times}$}
      \end{prooftree}

    \item a family of operators for all $\type{\sigma}, \type{\tau} \in A^{\Rightarrow,\times}$
      \begin{align}
        \abs^{\type{\sigma},\type{\tau}}
        &: \type{\sigma}.\type{\tau} \to (\type{\sigma \Rightarrow \tau})
        \tag{abstraction}\\
        \app^{\type{\sigma},\type{\tau}}
        &: (\type{\sigma \Rightarrow \tau}, \type{\sigma}) \to \type{\tau}
        \tag{application} \\
        \con{pair}^{\type{\sigma},\type{\tau}}
        &: (\type{\sigma}, \type{\tau}) \to \type{\sigma \times \tau}
        \tag{pair constructor} \\
        \con{fst}^{\type{\sigma},\type{\tau}}
        &: \type{\sigma \times \tau} \to \type{\sigma}
        \tag{first projection} \\
        \con{snd}^{\type{\sigma},\type{\tau}}
        &: \type{\sigma \times \tau} \to \type{\tau}
        \tag{second projection}
      \end{align}
  \end{enumerate}
\end{example}

\begin{example}
  Consider $\lambda$-term $(\lambda x. \m \; \langle f, x\rangle) : \type{\sigma \Rightarrow \tau}$. Assume $f : \type{\sigma \Rightarrow \tau}$. When representing this term in second-order syntax we have a choice:
  \begin{enumerate}
    \item if we treat $\m$ as a nullary metavariable, then
      \begin{align*}
        \m : []\type{((\sigma \Rightarrow \tau) \times \sigma) \Rightarrow \tau}
          \mid f : \type{\sigma \Rightarrow \tau}
          \vdash \con{abs}(x.\con{app}(\m[], \con{pair}(f, x))) : \type{\sigma \Rightarrow \tau}
      \end{align*}
    \item if we treat $\m$ as metavariable with one parameter, then
      \begin{align*}
        \m : [\type{(\sigma \Rightarrow \tau) \times \sigma}] \type{\tau}
          \mid f : \type{\sigma \Rightarrow \tau}
          \vdash \con{abs}(x.\m[\con{pair}(f, x)]) : \type{\sigma \Rightarrow \tau}
      \end{align*}
  \end{enumerate}
  These representations are equivalent and the next sections provide a formal notion of such an equivalence.
\end{example}

\begin{example}
  The equational presentation of the simply-typed $\lambda$-calculus with pairs:
  \begin{align*}
    \metavar{m}{} : [\type{\sigma}]\type{\tau}, \metavar{n}{} : []\type{\sigma}
      \mid \cdot
      &\vdash \con{app}(\con{abs}(x. \metavar{m}{}[x]), \metavar{n}{}[])
      \equiv \metavar{m}{}[\metavar{n}{}[]] : \type{\tau}
      \tag{$\beta_\lambda$} \\
    \metavar{m}{} : []\type{\sigma}, \metavar{n}{} : []\type{\tau}
      \mid \cdot
      &\vdash \con{fst}(\con{pair}(\metavar{m}{}[], \metavar{n}{}[]))
      \equiv \metavar{m}{}[] : \type{\sigma}
      \tag{$\beta_{\pi_1}$} \\
    \metavar{m}{} : []\type{\sigma}, \metavar{n}{} : []\type{\tau}
      \mid \cdot
      &\vdash \con{snd}(\con{pair}(\metavar{m}{}[], \metavar{n}{}[]))
      \equiv \metavar{n}{}[] : \type{\tau}
      \tag{$\beta_{\pi_2}$} \\
    \metavar{m}{} : []\type{\sigma \Rightarrow \tau}
      \mid \cdot
      &\vdash \con{abs}(x. \con{app}(\metavar{m}{}[], x))
      \equiv \metavar{m}{}[] : \type{\sigma \Rightarrow \tau}
      \tag{$\eta_\lambda$} \\
    \metavar{m}{} : []\type{\sigma \times \tau}
      \mid \cdot
      &\vdash \con{pair}(\con{fst}(\metavar{m}{}[]), \con{snd}(\metavar{m}{}[]))
      \equiv \metavar{m}{}[] : \type{\sigma \times \tau}
      \tag{$\eta_\times$}
  \end{align*}
\end{example}

See~\cref{figure:procedure-example} for an example of a search path generated by the $E$-unification procedure for the following problem in simply-typed $\lambda$-calculus:
  \[
    \exists \metavar{m}{} : [](\type{((\sigma \Rightarrow \tau) \times ((\sigma \Rightarrow \tau) \Rightarrow \tau)) \Rightarrow \tau}). \;
    (\forall g : \type{\sigma \Rightarrow \tau}, y : \type{\sigma}. \;
    \con{app}(\metavar{m}{}[], \con{pair}(g, \con{abs}(x. \con{app}(x, y))))
      \meq \con{app}(g, y) : \type{\tau})
  \]

\begin{figure}
  \begin{mdframed}
    \includegraphics{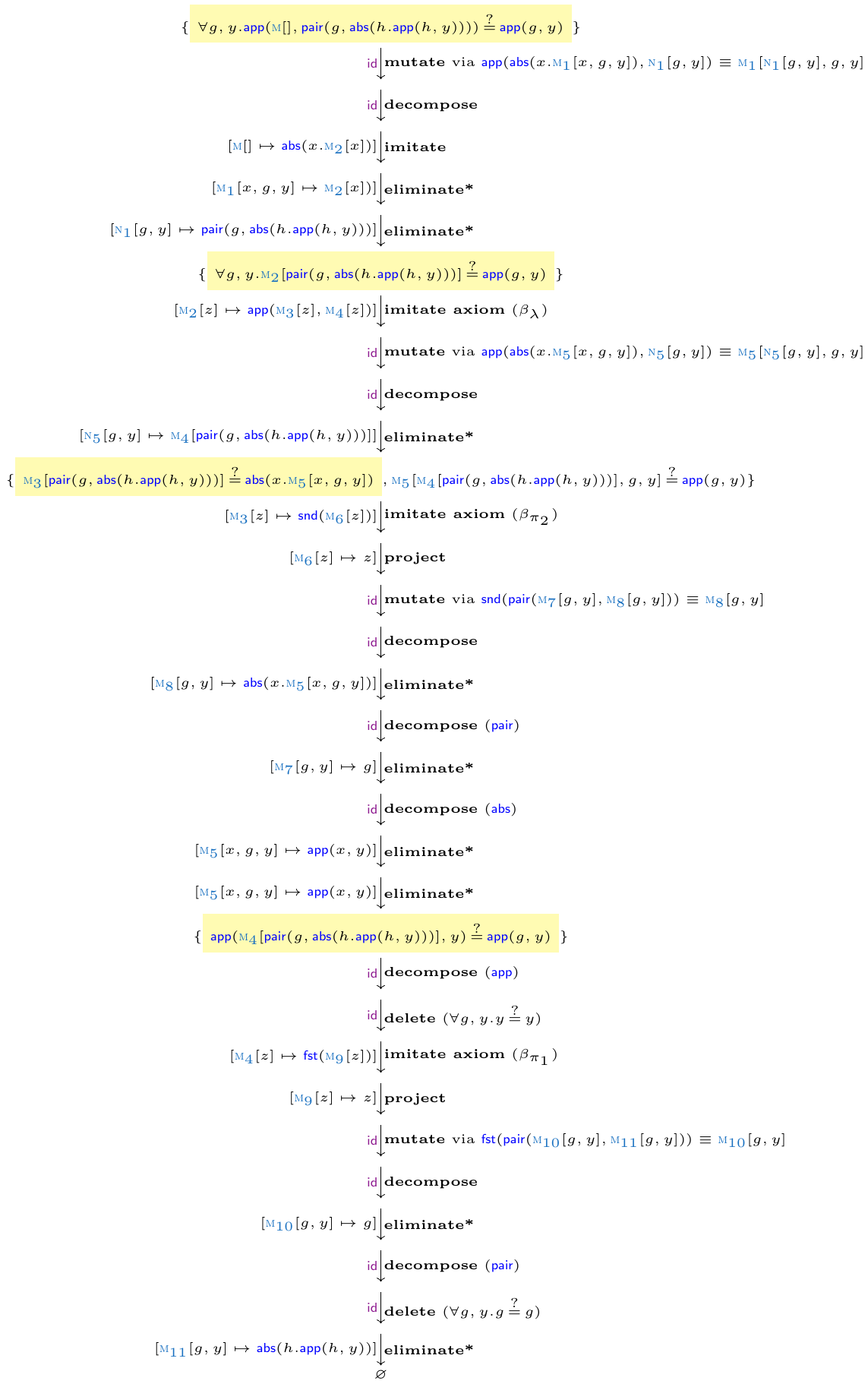}
  \end{mdframed}
  \caption{Sample search path, exhibiting many rules, for an $E$-unification problem in $\lambda$-calculus: $\exists \m : \type{((\sigma \Rightarrow \tau) \times \sigma) \Rightarrow \tau}. \; \forall g : \type{\sigma \Rightarrow \tau}, y : \type{\sigma}. \; \m \; \langle g, \lambda h. h \; y \rangle \meq g \; y : \type{\tau}$. Types are omitted to save space. Composition of substitutions yields $[\m[] \mapsto \con{abs}(x. \con{app}(\con{snd}(x), \con{fst}(x))), \ldots ]$.}
  \label{figure:procedure-example}
\end{figure}

\section{Proof of Completeness}
\label{section:completeness}

In this section we prove our main theorem, showing that our unification procedure is complete.

We start with a definition of mixed operators:

\begin{definition}
  \label{definition:mixed-operator}
  We say that an operator $\con{F} : (\type{\overline{\alpha_1}}.\type{\beta_1}, \ldots, \type{\overline{\alpha_n}}.\type{\beta_n}) \to \type{\gamma}$ is \defemph{mixed} iff $\type{\alpha_i}$ is empty and $\type{\alpha_j}$ is not empty for some $i$ and $j$.
\end{definition}

Dealing with mixed operators can be very non-trivial. In the following theorem, we assume that all operators either introduce scopes in all subterms, or in none. That is, for each operator $\con{F} : (\type{\overline{\alpha_1}}.\type{\beta_1}, \ldots, \type{\overline{\alpha_n}}.\type{\beta_n}) \to \type{\gamma}$, either $|\type{\overline{\alpha_i}}| = 0$ for all $i$ or $|\type{\overline{\alpha_i}}| > 0$ for all $i$. The assumption is justified since we can always encode a mixed operator as a combination of non-mixed operators. For example, $\con{let}(t_1, x.t_2)$ can be encoded as $\con{let}(t_1, \con{block}(x.t_2))$.

\begin{theorem}
  \label{theorem:completeness}
  Assuming no mixed operators are used, the procedure described in~\cref{definition:naive-unification} is complete, meaning that all paths from a root to all \textbf{\color{OliveGreen}(success)} leaves in the search tree constructed by the procedure, form a complete (but not necessarily minimal) set of $E$-unifiers. More specifically, let $E$ be an equational presentation and $\langle \Theta, S \rangle$ be an $E$-unification problem. Then for any $E$-unifier $\metasubst{\theta} \in U_E(S)$ there exists a path $S \stackrel{\metasubst{\xi}}{\longrightarrow} \varnothing$ such that $\metasubst{\xi} \preccurlyeq_E \metasubst{\theta}$.
\end{theorem}

\begin{remark}
  The unification procedure may produce redundant unifiers. For example, consider the following unification problem in simply-typed $\lambda$-calculus:
  \begin{align*}
    \exists \m : [\type{\sigma}, \type{\sigma}]\type{\tau}, \metavar{n}{} : [\type{\sigma}](\type{\sigma \Rightarrow \tau}). \;
    \forall x : \type{\sigma}. \;
    \m[x, x] \meq \con{app}(\metavar{n}{}[x], x)
  \end{align*}

  Depending on whether we start with the \textbf{(imitate)} rule or the \textbf{(mutate)} first, we can arrive at the following unifiers:
  \begin{align*}
    \metasubst{\theta_1} &= [\m[z_1, z_2] \mapsto \con{app}(\metavar{n}{}[z_1], z_2)] \\
    \metasubst{\theta_2} &= [\m[z_1, z_2] \mapsto \metavar{t}{}[z_2, z_1], \; \metavar{n}{}[z_1] \mapsto \con{abs}(z. \metavar{t}{}[z_1, z])]
  \end{align*}
  It is clear that $\metasubst{\theta_1} \not= \metasubst{\theta_2}$, but $\metasubst{\theta_1} \preccurlyeq_E \metasubst{\theta_2}$ (witnessed by $[\metavar{n}{}[z_1] \mapsto \con{abs}(z. \metavar{t}{}[z_1, z])]$). Hence, the set of $E$-unifiers produced for this unification problem is not minimal (by~\cref{definition:csu}).
\end{remark}

Our proof is essentially a combination of the two approaches: one by Gallier and Snyder in their proof of completeness for general (first-order) $E$-unification~\cite{GallierSnyder1989}, and another one by Jensen and Pietrzykowski (JP)~\cite{JensenPietrzykowski1976}, refined by Vukmirovic, Bentkamp, and Nummelin (VBN)~\cite{VukmirovicBentkampNummelin2021} for full higher-order unification. In particular, we need to reuse some of the ideas from the latter when dealing with parametrised metavariables. However, we cannot reuse the idea of JP's $\omega$-simplicity or VBN's base-simplicity, as those are dependent crucially on the $\eta$-long terms, $\lambda$-abstraction, function application, which are not accessible to us in a general second-order abstract syntax. Instead, we reuse the ideas of Gallier and Snyder to understand when it is okay to decompose terms. To understand when to apply \textbf{(iterate)} rule, we also look at the rewrite sequence instead of $\omega$-simplicity of terms.

The main idea of the proof is to take the unification problem $S$ together with its $E$-unifier $\metasubst{\theta}$ and then choose one of the rules of the procedure guided by $\metasubst{\theta}$. Applying a rule updates constraints and the remaining substitution is also updated. To show that this process terminates, we introduce a measure that strictly decreases with each rule application.

\begin{definition}
  \label{definition:measure}
  Let $\metasubst{\theta} \in U_E(S)$. Then, define the \defemph{measure} on pairs $\langle S, \metasubst{\theta} \rangle$ as the lexicographic comparison of
  \begin{enumerate}
    \item sum of lengths of the rewriting sequences $\metasubst{\theta}s \stackrel{*}{\longleftrightarrow}_E \metasubst{\theta}t$ for all $s \meq t$ of $S$;
    \item total number of operators used in $\metasubst{\theta}$;
    \item total number of metavariables used in $\metasubst{\theta}$;
    \item sum of sizes of terms in $S$.
  \end{enumerate}

  We denote the quadruple above as $\ord(S, \metasubst{\theta})$.
\end{definition}

The following definition helps us understand when we should apply the \textbf{(project)} rule:

\begin{definition}
  A metavariable $\m : [\type{\overline{\sigma}}]\type{\tau}$ is \defemph{projective at $j$ relative to} $\metasubst{\theta}$ if $\metasubst{\theta}\m[\overline{z}] = z_j$.
\end{definition}

One of the crucial points in the proof is to understand whether we can apply \textbf{(identify)} or \textbf{(eliminate)} rules for constraints with two metavariables on both sides. The following lemma provides precise conditions for this, allowing for \textbf{(identify)} when metavariables are distinct or \textbf{(eliminate)} when they are equal.

\begin{lemma}
  \label{lemma:identify}
  Let $s = \m[\overline{u}]$ and $t = \metavar{n}{}[\overline{v}]$
  such that $\metasubst{\zeta}s \stackrel{*}{\longleftrightarrow}_E \metasubst{\zeta}t$.
  Let $s_1, \ldots, s_n$ be the subterms of $\metasubst{\zeta}s$,
      $t_1, \ldots, t_n$ the subterms of $\metasubst{\zeta}t$
  such that the rewriting sequence
    $\metasubst{\zeta}s \stackrel{*}{\longleftrightarrow}_E \metasubst{\zeta}t$
  corresponds to the union of independent rewritings
    $s_i \stackrel{*}{\longleftrightarrow}_E t_i$ for all $i \in \{1, \ldots, n\}$.
  If for all $i$ we have
    either that $s_i$ is a subterm of an occurrence of $\metasubst{\zeta}u_{j_i}$
    or that     $t_i$ is a subterm of an occurrence of $\metasubst{\zeta}v_{j_i}$,
  then there exist terms $w$, $\overline{u'}$, $\overline{v'}$ such that
  \begin{enumerate}
    \item
      $\overline{u'}[\overline{y} \mapsto \metasubst{\zeta}\overline{v}]
      \equiv_E \metasubst{\zeta}\overline{u}$ and
      $\overline{v'}[\overline{z} \mapsto \metasubst{\zeta}\overline{u}]
      \equiv_E \metasubst{\zeta}\overline{v}$,
    \item
      $\metasubst{\zeta}\m[\overline{z}] = w[\overline{y} \mapsto \overline{v'}]$
      and
      $\metasubst{\zeta}\metavar{n}{}[\overline{y}] = w[\overline{z} \mapsto \overline{u'}]$.
  \end{enumerate}
\end{lemma}

\begin{proof}
  We define an auxiliary family of terms
  $\Xi \mid \overline{x} : \type{\overline{\alpha}} \vdash w'(l, r) : \type{\tau}$
  for pairs of terms
  $\Xi \mid \overline{x} : \type{\overline{\alpha}}, \overline{z} : \type{\overline{\gamma}} \vdash l : \type{\tau}$ and
  $\Xi \mid \overline{x} : \type{\overline{\alpha}}, \overline{y} : \type{\overline{\beta}} \vdash r : \type{\tau}$
  such that $l$ is a subterm of $\metasubst{\zeta}\m[\overline{z}]$
        and $r$ is a subterm of $\metasubst{\zeta}\metavar{n}{}[\overline{y}]$
  satisfying
  $l[\overline{z} \mapsto \metasubst{\zeta}\overline{u}]
  \equiv_E r[\overline{y} \mapsto \metasubst{\zeta}\overline{v}]$.
  We define $w'(l, r)$ inductively on the structure of $l$ and $r$, maintaining
  $l \equiv_E w'(l, r)[\overline{y} \mapsto \overline{v'}]$
  and
  $r \equiv_E w'(l, r)[\overline{z} \mapsto \overline{u'}]$:
  \begin{enumerate}
    \item if $l = x_i$ or $r = x_i$, then $l = r$ and $w'(l, r) = x_i$;
    \item if $l = z_k$ then $w'(l, r) = z_k$ and $u_k' = r$;
    \item if $r = y_k$ then $w'(l, r) = y_k$ and $v_k' = l$;
    \item if $l = \con{F}(\overline{\overline{a}.p})$ and $r = \con{F}(\overline{\overline{a}.q})$ then
      $w'(l, r) = \con{F}(\overline{\overline{a}.w'(p, q)})$; note that $w'(p_i,q_i)$ is well-defined for all $i$, since $l[\overline{z} \mapsto \metasubst{\zeta}\overline{u}] \equiv_E r[\overline{y} \mapsto \metasubst{\zeta} \overline{v}]$ implies component-wise equality $p_i[\overline{z} \mapsto \metasubst{\zeta} \overline{u}] \equiv_E q_i[\overline{y} \mapsto \metasubst{\zeta} \overline{v}]$ for all $i$. This is true, since otherwise we are rewriting (at root) both $l$ and $r$, but $l[\overline{z} \mapsto \metasubst{\zeta} \overline{u}]$ nor $r[\overline{y} \mapsto \metasubst{\zeta} \overline{v}]$ corresponds to a parameter occurrence $\metasubst{\zeta}u_{j}$ or $\metasubst{\zeta}v_j$ in terms $\metasubst{\zeta}s$ or $\metasubst{\zeta}t$ correspondingly.
    \item if $l = \m[\overline{p}]$ and $r = \m[\overline{q}]$ then
      $w'(l, r) = \m[\overline{w'(p, q)}]$; here, $w'(p_i, q_i)$ is well-defined for all $i$, similarly to the previous case.
  \end{enumerate}
  If $u_k$ (or $v_k$) has not been defined for some $k$, it means that a corresponding parameter is not essential and can be eliminated. We set\footnote{alternatively, we could have adjusted the statement to only mention subsets of variables $\overline{z}$ and $\overline{y}$ that are used at least once} such $u_k$ to be a fresh metavariable $\metavar{u}{_k}[]$.
  We set $w = w'(\metasubst{\zeta}\m[\overline{z}], \metasubst{\zeta}\metavar{n}{}[\overline{y}])$. By construction,
  $
    \metasubst{\zeta}\m[\overline{z}]
    = w[\overline{y} \mapsto \overline{v'}]
  $
  and
  $
    \metasubst{\zeta}\metavar{n}{}[\overline{y}]
    = w[\overline{z} \mapsto \overline{u'}]
  $.
\end{proof}

\begin{corollary}
  \label{lemma:eliminate}
  Let $s = \m[\overline{u}]$ and $t = \m[\overline{v}]$
  such that $\metasubst{\zeta}s \stackrel{*}{\longleftrightarrow}_E \metasubst{\zeta}t$.
  Let $s_1, \ldots, s_n$ be the subterms of $\metasubst{\zeta}s$,
      $t_1, \ldots, t_n$ the subterms of $\metasubst{\zeta}t$
  such that the rewriting sequence
    $\metasubst{\zeta}s \stackrel{*}{\longleftrightarrow}_E \metasubst{\zeta}t$
  corresponds to the union of independent rewritings
    $s_i \stackrel{*}{\longleftrightarrow}_E t_i$ for all $i \in \{1, \ldots, n\}$.
  If for all $i$ we have
    either that $s_i$ is a subterm of an occurrence of $\metasubst{\zeta}u_{j_i}$
    or that     $t_i$ is a subterm of an occurrence of $\metasubst{\zeta}v_{j_i}$,
  then, there exists a sequence $1 \leq j_1 < \ldots < j_k \leq $
  such that 
  and $FV(\metasubst{\zeta}\m[\overline{z}]) = \{ z_{j_1}, \ldots, z_{j_k} \}$
  and $u_{j_i} \equiv_E v_{j_i}$ for all $i$.
\end{corollary}

The following lemma will help us generalize solutions in~\cref{proof:generalize} of the proof below.

\begin{lemma}
  \label{lemma:subterm-no-mixed}
  Let $\Xi \mid \overline{x} : \type{\overline{\alpha}} \vdash w : \type{\sigma}$ be
  a subterm of $\Xi \mid \overline{x} : \type{\overline{\alpha}} \vdash t : \type{\tau}$.
  If $t$ does not contain mixed operators,
  then there exists a substitution
  $\metasubst{\zeta_{w, t}} = [\metavar{h}{}[z, \overline{y}] \mapsto h]$
  and a collection of terms $\Xi \mid \overline{x} : \type{\overline{\alpha}} \vdash \overline{s} : \overline{\type{\beta}}$,
  such that each $s_i$ is a subterm of $t$ and
  $\metasubst{\zeta_{w, t}}\metavar{h}{}[w, \overline{s}] = t$.
\end{lemma}
\begin{proof}
  Note that $w$ and $\overline{s}$ are subterms of $t$ and are not under binders (since they have the same variable context). Then, by induction on the structure of $t$:
  \begin{enumerate}
    \item if $t = w$, then $\metasubst{\zeta_{w, t}} = [\metavar{h}{}[z] \mapsto z]$;
    \item if $t = \con{F}(\overline{x_1}.t_1, \ldots, \overline{x_n}.t_n)$
      then, since $t$ does not contain mixed operators, $\overline{x_i}$ is empty for all $i$.
      Now, if $w$ is a subterm of $t_i$
      and $\metasubst{\zeta_{w, t_i}} = [\metavar{h}{_{t_i}}[z, y_1', \ldots, y_k'] \mapsto h_{t_i}]$
      then $\metasubst{\zeta} = [\metavar{h}{}[z, y_1', \ldots, y_k', y_1, \ldots, y_{n-1}] \mapsto \con{F}(y_1, \ldots, y_{i-1}, h_{t_i}, y_i, \ldots, y_{i-1})]$.
    \item if $t = \metavar{n}{}[t_1, \ldots, t_n]$ such that $w$ is a subterm of $t_i$
      and $\metasubst{\zeta_{w, t_i}} = [\metavar{h}{_{t_i}}[z, y_1', \ldots, y_k'] \mapsto h_{t_i}]$
      then $\metasubst{\zeta} = [\metavar{h}{}[z, y_1', \ldots, y_k', y_1, \ldots, y_{n-1}] \mapsto \metavar{n}{}[y_1, \ldots, y_{i-1}, h_{t_i}, y_i, \ldots, y_{i-1}]]$.
  \end{enumerate}
  Note that case of $t = x$ is impossible unless $t = w$ (case 1).
\end{proof}

  We are now ready to prove~\cref{theorem:completeness}.

\begin{proof}[Proof of \cref{theorem:completeness}]
  Let $S_0 = S$ and $\metasubst{\theta_0} = \metasubst{\rho \circ \theta}$, where $\metasubst{\rho}$ is some renaming substitution such that every metavariable occurring in $\metasubst{\theta_0} S_0$ does not occur in $S_0$. Note that $\metasubst{\theta_0}$ is an $E$-unifier of $S$, since $\metasubst{\theta}$ is by assumption.

  We now inductively define $S_i, \metasubst{\xi_i}$, and $\metasubst{\theta_i}$ until we reach some $i$ such that $S_i = \varnothing$. We ensure that $\ord(S_i, \metasubst{\theta_i})$ decreases with every step, so that such sequence of steps would always terminate. We maintain the following invariants for each step:
  \begin{enumerate}
    \item $\langle S_i, \metasubst{\theta_i} \rangle \stackrel{\metasubst{\xi_i}}{\longrightarrow} \langle S_{i+1}, \metasubst{\theta_{i+1}} \rangle$ where $S_i \stackrel{\metasubst{\xi_i}}{\longrightarrow} S_{i+1}$ by some rule of the unification procedure;
    \item $\ord(S_{i+1}, \metasubst{\theta_{i+1}}) < \ord(S_i, \metasubst{\theta_i})$;
    \item $\metasubst{\theta_i} \in U_E(S_i)$;
    \item $\metasubst{\theta_0} \equiv_E \metasubst{\theta_i \circ \xi_{i-1} \circ \ldots \circ \xi_0}$;
    \item every free variable occurring in $\metasubst{\theta_i} S_i$ does not occur in $S_i$;
  \end{enumerate}
  If $S_i \not= \varnothing$, then let $\forall \overline{x} : \type{\overline{\sigma}}. s \meq t : \type{\tau}$ be a constraint in $S_i$. We consider three major cases with respect to the rewriting sequence $\Theta_i \mid \overline{x} : \type{\overline{\sigma}} \vdash \metasubst{\theta_i} s \stackrel{*}{\longleftrightarrow}_E \metasubst{\theta_i} t : \type{\tau}$:
  \begin{enumerate}
    \item \textbf{The rewriting sequence contains a root rewrite.}
      More precisely, there exists a sequence $\metasubst{\theta_i}s = u_0 \longleftrightarrow_E \ldots \longleftrightarrow_E u_n = \metasubst{\theta_i}t$ and some term $u_j$ such that $u_j \longleftrightarrow_E u_{j+1}$ is a direct application of a rewrite using an axiom.
      This means that $s$ and $t$ can be unified by a direct use of an axiom. More specifically,
      there exists an instantiation $\metasubst{\zeta}$ of an axiom $\Xi \mid \cdots \vdash l \equiv r : \type{\tau}$ from $E$ such that $\metasubst{\zeta}l = u_{j}$, $u_{j+1} = \metasubst{\zeta}r$, and $\metasubst{\theta_i}$ unifies both $s \meq u_j$ and $u_{j+1} \meq t$. Thus, we can apply \textbf{(mutate)} rule.
      The measure decreases since the rewrite sequence $s \stackrel{*}{\longleftrightarrow}_E t$ is now split into two sequences $s \stackrel{*}{\longleftrightarrow} u_j$ and $u_{j+1} \stackrel{*}{\longleftrightarrow}_E t$ such that sum of lengths of new sequences is exactly one less than the length of the original sequence.

  \item \textbf{Rewriting sequence is empty or does not contain a root rewrite.}
    This means that rewrites may only happen in subterms.
      \begin{enumerate}
        \item If $s = x$ and $t = y$ where $x$ and $y$ are variables,
          then $x = y$ and we can apply \textbf{(delete)} rule, with $\metasubst{\xi_i} = \metasubst{\mathsf{id}}$ and $\metasubst{\theta_{i+1}} = \metasubst{\theta_i}$.
          The measure is reduced since the total size of constraints is reduced, while the rewriting sequences and the remaining substitution remain the same.

        \item \label{proof:decompose} If $s = \con{F}(\overline{\overline{z}.u})$ and $t = \con{F}(\overline{\overline{z}.v})$,
          then, $\metasubst{\theta_i}s = \con{F}(\overline{\overline{z}.\metasubst{\theta_i}u})$
          and $\metasubst{\theta_i}t = \con{F}(\overline{\overline{z}.\metasubst{\theta_i}v})$.
          Since there are no root rewrites, $\metasubst{\theta_i}$ unifies each pair $u_j \meq v_j$ in corresponding extended contexts,
          so we can apply \textbf{(decompose)} rule with $\metasubst{\xi_i} = \metasubst{\mathsf{id}}$ and $\metasubst{\theta_{i+1}} = \metasubst{\theta_i}$.
          Note that the chain of rewrites may be split into several chains, but the total sum of lengths remains the same. Second component of the measure also remains unchanged. We reduce the third component of the measure, since the total size of terms in the unification problem decreases, the sum of chains of rewrites is unchanged.

        \item \label{proof:project} If $s = \m[\overline{u}]$ and $\m$ is projective at $j$ relative to $\metasubst{\theta_i}$ then we can apply \textbf{(project)} rule with $\metasubst{\xi_i} = [\m[\overline{z}] \mapsto z_j]$.
          Note that the chain of rewrites remains unchanged and $\metasubst{\xi_i}$ does not take any operators away from $\metasubst{\theta_{i+1}}$ (which is a restriction of $\metasubst{\theta_i}$ to metavariables other than $\m$). We reduce the measure by reducing the total size of terms in the unification problem.

        \item \label{proof:imitate} If $s = \m[\overline{u}]$ where $\m$ is not projective relative to $\metasubst{\theta_i}$
          and $\metasubst{\theta_i}s = \con{F}(\overline{\overline{z}.u})$
          and $t = \con{F}(\overline{\overline{z}.v})$,
          then $\metasubst{\theta_{i}}$ unifies each pair $u_i \meq \metasubst{\zeta}v_i$ in corresponding extended contexts
          and we can apply \textbf{(imitate)} rule
          with $\metasubst{\xi_i} = [\m[\overline{z}] \mapsto \con{F}(\overline{\overline{x}.\metavar{t}{}[\overline{z},\overline{x}]})]$.
          Let $\metasubst{\theta_i}\m[\overline{z}] = \con{F}(\overline{\overline{x}.w})$, then $\metasubst{\theta_{i+1}}$ is constructed from $\metasubst{\theta_i}$ by removing mapping for $\m$ and adding mappings $[\metavar{t}{_j}[\overline{z},\overline{x_j}] \mapsto w_j]$ for all $j$.
          The chain of rewrites is unchanged and the measure decreases since we reduce the number of operators in $\metasubst{\theta_{i+1}}$.

        \item \label{proof:iterate} If $s = \m[\overline{u}]$ where $\m$ is not projective relative to $\metasubst{\theta_i}$ and $\metasubst{\theta_i}s \stackrel{*}{\longleftrightarrow}_E \metasubst{\theta_i}t$ contains a rewrite of a subterm $w$ in $\metasubst{\theta_i}s$ that is not a subterm of an occurrence of $\metasubst{\theta_i}u_i$ for some $i$, then 

          \begin{enumerate}
            \item If $w$ is under a binder in $\metasubst{\theta_i}s$, we take the outermost operator $\con{F}$ that binds a variable captured by $w$ (that is, $\metasubst{\theta_i}s = \ldots \con{F}(\overline{y_1}.s_1, \overline{y_j}. \ldots w \ldots, \overline{y_n}. s_n)$) and apply \textbf{(iterate)} rule
              with $\metasubst{\xi_i} = [\m[\overline{z}] \mapsto \metavar{m}{'}[\overline{z}, \con{F}(\overline{y_1}.\metavar{m}{_1}[\overline{z}, \overline{y_1}], \ldots, \overline{y_n}.\metavar{m}{_n}[\overline{z}, \overline{y_n}])]$.
              Let $\metasubst{\theta_i}\m[\overline{z}] = \con{F}(\overline{\overline{y}.s'})$,
              then $\metasubst{\theta_{i+1}}$ is defined as $\metasubst{\theta_i}$
              with mapping of $\m$ removed and added mappings
              $[\metavar{m}{_i}[\overline{z}, \overline{y_i}] \mapsto s_i]^{\{i \in \{1, \ldots, n\}\}}$.
              The chain of rewrites remains unchanged and the number of operators in $\metasubst{\theta_{i+1}}$ decreases by one, so the measure decreases.

            \item If $w = \con{F}(\ldots)$ and is not under a binder, then we apply \textbf{(iterate)} rule
              with $\metasubst{\xi_i} = [\m[\overline{z}] \mapsto \metavar{m}{'}[\overline{z}, \con{F}(\overline{y_1}.\metavar{m}{_1}[\overline{z}, \overline{y_1}], \ldots, \overline{y_n}.\metavar{m}{_n}[\overline{z}, \overline{y_n}])]$.
              We define $\metasubst{\theta_{i+1}}$ and show that the measure decreases analogously to the previous case.
            \item \label{proof:generalize} If $w = \metavar{w}{}[\overline{v}]$ and is not under a binder,
              then $\metasubst{\theta_i}\m[\overline{z}]$ contains $w' = \metavar{w}{}[\overline{v'}]$
              as a subterm
              and $v'_i[\overline{z} \mapsto \overline{u}] = v_i$ for all $i$
              (this is because $w$ is not a subterm of any of the $\metasubst{\theta_i}u_j$).
              Since $w'$ is also not under binder,
              then by~\cref{lemma:subterm-no-mixed} and assumption of no mixed operators
              we have that there exist terms $h$, $\overline{s}$, and a substitution
              $\metasubst{\zeta} = [\metavar{h}{}[\overline{z}, \overline{y}] \mapsto h]$
              such that $\metasubst{\zeta}h[\overline{v'}, \overline{s}] = \metasubst{\theta_i}\m[\overline{z}]$.
              Set $\theta_i' = [\m[\overline{z}] \mapsto h[\overline{y} \mapsto \overline{s}]]$.
              We have $\metasubst{\theta_i} = \metasubst{\zeta \circ \theta_i'}$,
              that is $\theta_i'$ is more general that $\theta_i$ modulo $E$.
              The rewriting sequence remains unchanged.
              If $\metasubst{\theta_i}s$ has an operator at root,
              then $\metasubst{\theta_i'}$ has fewer operators which decreases the measure.
              If $\metasubst{\theta_i}s$ has a metavariable at root,
              then $\metasubst{\theta_i'}$ has fewer metavariables which decreases the measure.
          \end{enumerate}

        \item \label{proof:identify}
  If $s = \m[\overline{u}]$ and $t = \metavar{n}{}[\overline{v}]$
  where $\m \not= \metavar{n}{}$, both $\m$ and $\metavar{n}{}$
  are not projective relative to $\metasubst{\theta_i}$
  and $\metasubst{\theta_i}s \stackrel{*}{\longleftrightarrow}_E \metasubst{\theta_i}t$
  corresponds to the union of independent rewritings
    $s_i \stackrel{*}{\longleftrightarrow}_E t_i$ for all $i \in \{1, \ldots, n\}$
  such that for all $i$ we have
    either that $s_i$ is a subterm of an occurrence of $\metasubst{\zeta}u_{j_i}$
    or that     $t_i$ is a subterm of an occurrence of $\metasubst{\zeta}v_{j_i}$,
    then by~\cref{lemma:identify} there exist terms $w$, $\overline{u'}$, $\overline{v'}$ such that
    $\metasubst{\theta_i}\m[\overline{z}] = w[\overline{y} \mapsto \overline{v'}]$
    and
    $\metasubst{\theta_i}\metavar{n}{}[\overline{y}] = w[\overline{z} \mapsto \overline{u'}]$.
    We now can apply \textbf{(identify)} rule with
    $\xi_i = [\m[\overline{z}] \mapsto \metavar{w}{}[\overline{y} \mapsto \overline{\metavar{y}{}[\overline{z}]}]
    \metavar{n}{}[\overline{y}] \mapsto \metavar{w}{}[\overline{z} \mapsto \overline{\metavar{z}{}[\overline{y}]}]]$.
   We define $\metasubst{\theta_{i+1}}$ to be defined as $\metasubst{\theta_i}$ without mappings for $\m$ and $\metavar{n}{}$ and with added mappings
   $[\metavar{w}{}[\overline{z}, \overline{y}] \mapsto w,
     \overline{\metavar{y}{}[\overline{z}]} \mapsto \overline{v'},
     \overline{\metavar{z}{}[\overline{y}]} \mapsto \overline{u'}]$.
     The rewriting sequence remains unchanged, but by~\cref{lemma:identify} the term $w$ is not a variable, so there is an operator or a metavariable that was mentioned twice in $\metasubst{\theta_i}$ (once for $\m$ and once for $\metavar{n}{}$) and is now mentioned once in $\metasubst{\theta_{i+1}}$ (for $\metavar{w}{}$), so the number of operators or metavariables in $\metasubst{\theta_{i+1}}$ decreases by at least 1. Thus, the measure decreases.

        \item \label{proof:eliminate}
  If $s = \m[\overline{u}]$ and $t = \m[\overline{v}]$
  where $\m$ is not projective relative to $\metasubst{\theta_i}$
  and $\metasubst{\theta_i}s \stackrel{*}{\longleftrightarrow}_E \metasubst{\theta_i}t$
  corresponds to the union of independent rewritings
    $s_i \stackrel{*}{\longleftrightarrow}_E t_i$ for all $i \in \{1, \ldots, n\}$
  such that for all $i$ we have
    either that $s_i$ is a subterm of an occurrence of $\metasubst{\zeta}u_{j_i}$
    or that     $t_i$ is a subterm of an occurrence of $\metasubst{\zeta}v_{j_i}$,
    then by~\cref{lemma:eliminate}
    we have $\overline{z'} = FV(\metasubst{\theta_i}\m[\overline{z}]) \subseteq \overline{z}$
    such that for each $z_k \in \overline{z'}$ we have $\metasubst{\theta_i} u_k \equiv_E \metasubst{\theta_i} v_k$. Consider two subcases:
    \begin{enumerate}

      \item If $\overline{z'} = \overline{z}$ we apply \textbf{(decompose)} rule
        with $\xi_i = \metasubst{\mathsf{id}}$ and $\metasubst{\theta_{i+1}} = \metasubst{\theta_i}$.
        The chain of rewrites remains, the remaining substitution is unchanged, but the total size of
        constraints is reduced, so the measure decreases.

      \item If $\overline{z'} \subset \overline{z}$ we apply \textbf{(eliminate)} rule
        with $\metasubst{\xi_i} = [\m[\overline{z}] \mapsto \metavar{e}{}[\overline{z'}]]$
        and $\metasubst{\theta_{i+1}}$ defined as a version of $\metasubst{\theta_i}$ with removed mapping for $\m$ and added mapping $[\metavar{e}{}[\overline{z'}] \mapsto \metasubst{\theta_i}s]$.
        Note that the chain of rewrites remains unchanged and $\metasubst{\xi_i}$ does not take any operators away from $\metasubst{\theta_{i+1}}$. We reduce the measure by reducing the total size of terms in the unification problem (as at least one parameter is removed from at least one metavariable $\m$).

    \end{enumerate}

      \end{enumerate}

  \end{enumerate}

  We now have a sequence $\langle S_0, \metasubst{\theta_0} \rangle \stackrel{\metasubst{\xi_0}}{\longrightarrow} \langle S_1, \metasubst{\theta_1} \rangle \stackrel{\metasubst{\xi_1}}{\longrightarrow} \ldots$. The sequence is finite since the measure $\ord(S_i, \metasubst{\theta_i})$ strictly decreases with every step. Therefore, $\langle S, \metasubst{\theta_0} \rangle \stackrel{\metasubst{\xi_0}}{\longrightarrow} \ldots \stackrel{\metasubst{\xi_n}}{\longrightarrow} \langle \varnothing, \metasubst{\mathsf{id}} \rangle$ and $\metasubst{\theta} \equiv_E \metasubst{\rho^{-1} \circ \theta_i \circ \xi_{i-1} \circ \ldots \xi_0} \equiv_E \metasubst{\rho^{-1} \circ \xi_n \circ \ldots \circ \xi_0} \preccurlyeq_E \metasubst{\xi_n \circ \ldots \circ \xi_0}$, completing the proof.
\end{proof}

\section{Discussion}
\label{section:discussion}

A pragmatic implementation of our procedure may enjoy the following changes.
We find that these help make a reasonable compromise between completeness and performance:
\begin{enumerate}
  \item remove \textbf{(iterate)} rule; this rule sacrifices completeness, but helps significantly reduce non-determinism; the solutions lost are also often highly non-trivial and might be unwanted in certain applications such as type inference;
  \item implement \textbf{(eliminate*)} rule;
  \item split axioms $E = B \uplus R$ such that $R$ constitutes a confluent and terminating term rewriting system, and introduce \textbf{(normalize)} rule to normalize terms (lazily) before applying any other rules except \textbf{(delete)} and \textbf{(eliminate*)};
  \item introduce a limit on a number of applications of \textbf{(mutate)} rule;
  \item introduce a limit on a number of bindings that do not decrease problem size;
  \item introduce a limit on total number of bindings.
\end{enumerate}

When adapting ideas from classical $E$-unification and HOU, some technical difficulties arise from having binders lacking (in general) the nice syntactic properties of $\lambda$-calculus. These difficulties affect both the design of our unification procedure, leading to some simplifications, and the completeness proof, requiring us to find a different approach to define the measure and consider cases that do not have analogues.

In the procedure, we had to simplify whenever those ideas relied on normalisation, $\eta$-expansion, or specific syntax of $\lambda$-terms. Many HOU algorithms look at syntactic properties of terms to determine which rules to apply. In particular, HOU algorithms often distinguish \emph{flex} and \emph{rigid} terms \cite{Huet1975, Miller1991}. Jensen and Pietrzykowski introduce a notion of \emph{$\omega$-simple} terms \cite{JensenPietrzykowski1976}. Vukmirovic, Bentkamp, and Nummelin~\cite{VukmirovicBentkampNummelin2021} introduce notions of \emph{base-simple} and \emph{solid} terms. These properties crucially depend on specific normalisation properties of $\lambda$-calculus, which might be inaccessible in an arbitrary second-order equational theory. Thus, our procedure contains more non-determinism than might be necessary.

One notable example of such simplication is in the imitation and projection bindings. In HOU algorithms, it is common to have substitutions of the form
\[
  [\m \mapsto \lambda x_1, \ldots, x_n. f \; (\metavar{h}{_1} \; x_1 \; \ldots \; x_n) \; \ldots \; (\metavar{h}{_k} \; x_1 \; \ldots \; x_n)]
\]
where $f$ can be a bound variable (one of $x_1, \ldots, x_n$) or a constant of type $\type{\sigma_1 \Rightarrow \ldots \Rightarrow \sigma_k \Rightarrow \tau}$. These are called Huet-style projection or imitation bindings~\cite{JensenPietrzykowski1976,VukmirovicBentkampNummelin2021} or partial bindings~\cite{Huet1975,Miller1991}. Huet-style projections (and conditions prompting their use) are non-trivial to generalise well to arbitrary second-order abstract syntax, so we skipped them in this paper, opting out for simpler rules but larger search space.

In the completeness proof for HOU algorithms, the syntactic properties of $\lambda$-calculus are heavily exploited. Their inaccessibility in a general second-order equational theory has contributed to some difficulties when developing the proof of completeness in \cref{theorem:completeness}. Perhaps, the most challenging of all was handling of the \cref{proof:generalize} of the proof which requires the assumption of no mixed operators and \cref{lemma:subterm-no-mixed}. These do not appear to have an analogue in completeness proofs for HOU or first-order E-unification.

\section{Conclusion and Future Work}

We have formulated the equational unification problem for second-order abstract syntax, allowing to reason naturally about unification of terms in languages with binding constructions. Such languages include, but are not limited to higher-order systems such as $\lambda$-calculus, which expands potential applications to more languages. We also presented a procedure to solve such problems and our main result shows completeness of this procedure.

In future work, we will focus on optimisations and recognition of decidable fragments of $E$-unification over second-order equations.

One notable optimisation is splitting $E$ into two sets $R \uplus B$, where $R$ is a set of directed equations, forming a confluent second-order term rewriting system, and $B$ is a set of undirected equations (such as associativity and commutativity axioms).

Another potential optimisation stems from a generalisation of Huet-style binding (also known as \emph{partial binding}), which can lead to more informed decisions on which rule to apply in the procedure, introduce Huet-style version of \textbf{(project)} and improve \textbf{(iterate)} rule, significantly reducing the search space. A version of such an optimisation has been implemented in a form of a heuristic to combine \textbf{(imitate)} and \textbf{(project)} rules by Kudasov~\cite{Kudasov2022}.

There are several well-studied fragments both for $E$-unification and higher-order unification. For example, unification in monoidal theories is essentially solving linear equations over semirings~\cite{Nutt1992}. In higher-order unification, there are several well-known decidable fragments such as pattern unification~\cite{Miller1991}. Vukmirovic, Bentkamp, and Nummelin have identified some of the practically important decidable fragments as well as a new one in their recent work on efficient full higher-order unification~\cite{VukmirovicBentkampNummelin2021}. It is interesting to see if these fragments could be generalised to second-order abstract syntax and used as oracles, possibly yielding an efficient $E$-unification for second-order abstract syntax as a strict generalisation of their procedure.

\bibliography{ms}


\end{document}